\documentclass[pra,twocolumn,amsmath,amssymb,showpacs, superscriptaddress,notitlepage,nofootinbib,onecolumn]{revtex4-2}

\usepackage{graphicx}
\usepackage[caption=false,labelformat=simple]{subfig}

\usepackage{siunitx,booktabs}
\usepackage{dcolumn}
\usepackage{amsmath,amssymb,mathtools,physics}
\usepackage[utf8]{inputenc}

\usepackage[hidelinks]{hyperref}
\usepackage{footnote}
\usepackage{enumitem}
\hypersetup{
	colorlinks   = true, 
	urlcolor     = blue, 
	linkcolor    = blue, 
	citecolor   = blue 
}
\usepackage[capitalize]{cleveref}

\usepackage [english]{babel}
\usepackage [autostyle, english = american]{csquotes} \MakeOuterQuote{"}

\usepackage[dvipsnames]{xcolor}

\usepackage{amsthm}

\newtheorem*{lemma}{Lemma}

\newcommand{\sdpcitations}{\cite{colesNumericalApproach2016,winickReliableNumerical2018,wangCharacterisingCorrelations2019,primaatmajaVersatileSecurity2019,huRobustInterior2022,araujoQuantumKey2023,zhouNumericalMethod2022}}

\newcommand{\sdpcitationsasymp}{\cite{colesNumericalApproach2016,winickReliableNumerical2018,wangCharacterisingCorrelations2019,primaatmajaVersatileSecurity2019,huRobustInterior2022,araujoQuantumKey2023}}

\newcommand{\affvqcc}{Vigo Quantum Communication Center, University of Vigo, Vigo E-{36310}, Spain}
\newcommand{\affuvigo}{Escuela de Ingeniería de Telecomunicación, Department of Signal Theory and Communications, University of Vigo, Vigo E-36310, Spain}
\newcommand{\affatlantic}{atlanTTic Research Center, University of Vigo, Vigo E-36310, Spain}

\newcommand{\affgeneva}{Group of Applied Physics, University of Geneva, 1205 Geneva, Switzerland}

\begin{document}
	
	\setlength{\parskip}{3pt}
	\setlength{\parindent}{0pt}
	\title{Numerical security analysis for quantum key distribution with partial state characterization}
	
	\author{Guillermo Currás-Lorenzo} \email{gcurras@vqcc.uvigo.es}
	\affiliation{\affvqcc} \affiliation{\affuvigo} \affiliation{\affatlantic} 
	\author{Álvaro Navarrete}
	\affiliation{\affvqcc} \affiliation{\affuvigo} \affiliation{\affatlantic} 
	\author{Javier Núñez-Bon}
	\affiliation{\affvqcc}
	\affiliation{\affuvigo}
	\affiliation{\affgeneva}
	\author{Margarida Pereira}
	\affiliation{\affvqcc} \affiliation{\affuvigo} \affiliation{\affatlantic} 
	\author{Marcos Curty}
	\affiliation{\affvqcc} \affiliation{\affuvigo} \affiliation{\affatlantic} 
	\quad
	
	\begin{abstract}
		Numerical security proofs offer a versatile approach for evaluating the secret-key generation rate of quantum key distribution (QKD) protocols. However, existing methods typically require perfect source characterization, which is unrealistic in practice due to the presence of inevitable encoding imperfections and side channels. In this paper, we introduce a novel security proof technique based on semidefinite programming that can evaluate the secret-key rate for both prepare-and-measure and measurement-device-independent QKD protocols when only partial information about the emitted states is available, significantly improving the applicability and practical relevance compared to existing numerical techniques. We demonstrate that our method can outperform current analytical approaches addressing partial state characterization in terms of achievable secret-key rates, particularly for protocols with non-qubit encoding spaces. This represents a significant step towards bridging the gap between theoretical security proofs and practical QKD implementations.
	\end{abstract}

	\maketitle

	\section{Introduction}
	
	In theory, quantum key distribution (QKD) enables two parties, Alice and Bob, to establish a shared secret key with information-theoretic security, ensuring its confidentiality even in the presence of an eavesdropper, Eve, with unlimited computational resources. However, in practice, the security of QKD implementations can be compromised due to discrepancies between the assumptions made in theoretical security proofs and the actual behavior of the devices. For example, standard security proofs of the BB84 protocol assume that Alice emits states drawn from two mutually unbiased qubit bases and that no setting-choice information is leaked from her lab. In practice, these ideal conditions are unattainable due to inevitable encoding imperfections and side channels. Rigorously accounting for such source imperfections is essential to ensure the security of practical QKD implementations \cite{ImplementationAttacksBSI}. 
	
	Over the past two decades, analytical security proofs have made significant progress in addressing this crucial problem \cite{gottesmanSecurityQuantum2004,loSecurityQuantum2007,tamakiLosstolerantQuantum2014,pereiraQuantumKey2019,pereiraModifiedBB842023,pereiraQuantumKey2020,curras-lorenzoSecurityFramework2023}. In particular, notably, Ref.~\cite{curras-lorenzoSecurityFramework2023} has shown that the finite-key security of QKD against coherent attacks can be guaranteed without requiring a detailed characterization of the emitted states. Namely, this analysis only needs the assumption that
	\begin{equation}	\label{eq:emitted_state_assumption}
		\abs{\braket{\phi_j}{\psi_{j}}}^{{2}} \geq 1 - \epsilon_j,
	\end{equation}
	where $\ket{\psi_j}$ is the actual state emitted by Alice when she selects setting $j$ (e.g., for BB84, $j \in \{0,1,+,-\}$), $\ket{\phi_j}$ is an arbitrary reference state, and $\epsilon_j$ bounds their deviation. This partial characterization requirement is very useful because certain imperfections, such as the states of side channels leaking information about Alice's setting $j$, could in principle live in arbitrarily-dimensional spaces, making it extremely challenging, if not impossible, to fully characterize them in practice. These imperfections can be incorporated into the analysis through the fidelity bound $\epsilon_j$, while easier-to-characterize imperfections, such as state preparation flaws (SPFs)---i.e., flaws within the encoding degree of freedom---can be incorporated into the definition of the reference states $\ket{\phi_j}$.
	
	On the other hand, numerical security analyses \sdpcitations{} based on convex optimization techniques have recently emerged as a flexible alternative to analytical proofs. By eliminating the need to redo the entire analysis for each new protocol variation, they enable quick and efficient evaluations across a broad range of QKD scenarios. Moreover, they can often yield optimal or near-optimal asymptotic key rates, even for scenarios lacking the symmetries that analytical proofs typically rely on. In principle, this flexibility should make them well-suited for evaluating the key-rate impact of various types of encoding imperfections---which often break such symmetries---and to assess the tightness of analytical results. Unfortunately, however, existing numerical analyses \sdpcitations{} require a complete description of the emitted states, including all imperfections. This unrealistic demand significantly limits their practical applicability.

	In this paper, we close this gap by introducing a novel security proof approach based on semidefinite programming (SDP)  that requires only partial knowledge of the emitted states, i.e., a bound of the form in \cref{eq:emitted_state_assumption}. Our approach can be applied to both prepare-and-measure (P\&M) and measurement-device-independent (MDI) protocols \cite{loMeasurementDeviceIndependentQuantum2012}, including decoy-state protocols \cite{loDecoyState2005,maPracticalDecoy2005,hwangQuantumKey2003,wangBeatingPhotonNumberSplitting2005}. Our numerical simulations show that, for standard BB84-like protocols, our approach only offers a very modest secret-key rate improvement with respect to the analytical result in \cite{curras-lorenzoSecurityFramework2023}, which suggests that the latter is close to optimal. However, for protocols with non-qubit encoding spaces, such as the coherent-light-based MDI protocol introduced in Ref.~\cite{navarretePracticalQuantum2021}, our method can yield substantially improved secret-key rates.

	\section{Main result}
	\label{sec:main_result}
	\label{sec:PM}
	
	{In this section, we present our security analysis  framework based on SDP. While our approach builds upon the foundation established by \cite{primaatmajaVersatileSecurity2019}, it introduces three key advances that significantly expand its applicability:
		
		\begin{enumerate}
			\item It is applicable to both P$\&$M and MDI-type protocols, whereas \cite{primaatmajaVersatileSecurity2019} is limited to MDI-type protocols.
			\item It requires only partial characterization of the source states (see condition in \cref{eq:emitted_state_assumption}), while \cite{primaatmajaVersatileSecurity2019} assumes full source characterization. 
			
			\item It can accommodate sources emitting mixtures of infinitely many pure states, enabling its application to decoy-state protocols \cite{loDecoyState2005,maPracticalDecoy2005,hwangQuantumKey2003,wangBeatingPhotonNumberSplitting2005}. In contrast, \cite{primaatmajaVersatileSecurity2019} is limited to discretely-modulated scenarios (i.e., scenarios in which the source emits a finite set of pure states), and therefore not applicable to standard decoy-state protocols. 
		\end{enumerate}
		
		Note that Point (2) is particularly critical, since demanding a full source characterization is extremely unrealistic, as highlighted in the Introduction. In the following, we demonstrate our analysis for discretely-modulated P$\&$M protocols; for extensions to MDI-type protocols and to non-discretely-modulated protocols, see \cref{app:MDI,app:decoy_state}, respectively. 
		
		\subsection*{Discretely-modulated P\&M protocols}
		
		In a discretely-modulated P$\&$M protocol, each round proceeds as follows: (1) Alice probabilistically selects a state from the set $\{\ket{\psi_j}_a
		\}_j$, where $j \in \{0,1,...\}$ refers to her setting choice; (2) she sends this state to Bob through an untrusted quantum channel; and (3) Bob performs a measurement on the incoming state. As the main focus of this work are source imperfections, for simplicity, we consider that Bob's measurement setup does not suffer from side channels and that it satisfies the basis-independent detection efficiency condition. That is, it can be described by two positive operator-valued measures (POVMs) with the same probability of outputting a conclusive bit outcome. Formally, we can express Bob's POVMs as 
		\begin{equation}
			\mathcal{Z} \eqqcolon \{\Gamma_{0_Z},\Gamma_{1_Z},\Gamma_{f}\}  \textrm{   and    } \mathcal{X} \eqqcolon \{\Gamma_{0_X},\Gamma_{1_X},\Gamma_{f}\},
		\end{equation}
		where $\Gamma_{b_Z}$ ($\Gamma_{b_X}$) is the operator associated to bit outcome $b \in \{0,1\}$ within POVM $\mathcal{Z}$ ($\mathcal{X}$) and $\Gamma_{f}$ is the operator associated to an inconclusive (i.e., failed) detection and is considered to be identical for both POVMs. While the basis-independent detection efficiency assumption simplifies our analysis, it could be relaxed by combining our approach with a recent result that allows for bounded detection efficiency mismatches \cite{tupkaryPhaseError2024}, or it could be eliminated entirely by considering an MDI-type protocol.

		For simplicity, we study the asymptotic regime and collective attacks. However, we remark that our proof could be lifted to general attacks by applying established techniques, as discussed in the Conclusion section. Eve's collective attack can be described as a quantum channel $\Lambda$ acting separately on each of Alice's emitted signals. Let $\{K_l\}_l$ be the set of Kraus operators of the operator-sum representation \cite{nielsenQuantumComputation2011} for the channel $\Lambda$. The probability $Y_{j}^{\gamma}$ that Bob obtains the measurement result $\gamma \in \{0_X, 1_X, f\}$ conditioned on him measuring $\mathcal{X}$ and Alice emitting the state $\ket{\psi_j}_a$ can be expressed as
		\begin{equation}
			\label{eq:Y_j_gamma}
			\begin{aligned}         
				Y_j^\gamma &= \Tr\big[\Lambda(\ketbra{\psi_j}_a) \Gamma_\gamma \big] = \Tr[\sum_{l}  K_l \ketbra{\psi_j}_a  K^{\dagger}_l \Gamma_\gamma] =  \sum_{l} \Tr[ K_l \ketbra{\psi_j}_a  K^{\dagger}_l \Gamma_\gamma ] \\
				&= \sum_{l} \Tr[\ketbra{\psi_j}_a  K^{\dagger}_l \Gamma_\gamma K_l ] = \Tr[\ketbra{\psi_j}_a \sum_{l}  K^{\dagger}_l \Gamma_\gamma K_l ] = \Tr[\ketbra{\psi_j}_a  \hat E_{\gamma} ] = \ev{\hat E_{\gamma}}{\psi_j}_a,
			\end{aligned}
		\end{equation}
		where in the third and fifth equalities we have used the linearity of the trace, in the fourth equality we have used the cyclic property of the trace, and in the second-to-last equality we have defined 
		\begin{equation}
			\hat E_{\gamma} = \sum_l K^{\dagger}_l \Gamma_\gamma K_l \geq 0.
		\end{equation}
		Note that
		\begin{equation}
			\label{eq:E_gamma_sum_less_I}
			\hat E_{0_X} + \hat E_{1_X} + \hat E_{f} =  \sum_l K^{\dagger}_l (\Gamma_{0_X} + \Gamma_{1_X} + \Gamma_{f}) K_l = \sum_l K^{\dagger}_l  K_l = \mathbb{I}.
		\end{equation}
		Therefore, we can regard $\mathcal{X}_{\rm eff} \coloneqq \{\hat E_{0_X}, \hat E_{1_X}, \hat E_{f}\}$ as a POVM that effectively combines both Eve's action and Bob's actual $\mathcal{X}$ measurement.
		
		For concreteness, let us assume that Alice and Bob obtain their sifted key from the detected events in which Alice emits $\ket{\psi_0}_a$ or $\ket{\psi_1}_a$ and Bob measures $\mathcal{Z}$. Moreover, let us assume that Alice emits these two states equiprobabilistically. Note that, in these rounds, Alice could have equivalently generated the entangled state\footnote{More generally, she could have generated the state
			\begin{equation}
				\label{eq:Psi_Z_def_alt}
				\ket{\Psi_Z}_{Aa} = \frac{1}{\sqrt 2}  \left(\ket{0}_A \ket{\psi_0}_a + e^{i \omega }\ket{1}_A \ket{\psi_1}_a  \right),
			\end{equation}
			for any $\omega \in [0,2\pi)$, where the phase $\omega$ is a free parameter of the analysis that affects the definition of the phase-error probability in \cref{eq:Y_ph_def_PM}. In this work, we have assumed $\omega = 0$, which is the optimal choice for all the scenarios that we have simulated. We remark, however, that in certain cases choosing a non-zero $\omega$ may be advantageous. This is the case, for instance, for BB84-type scenarios in which the emitted states are not in the \textit{XZ} plane \cite{curras-lorenzoFinitekeyAnalysis2021,curras-lorenzoSecurityFramework2023}. To apply our analysis for a non-zero $\omega$, one should simply update the definition of the phase-error probability in \cref{eq:Y_ph_def_PM,eq:Y_ph_def_PM_expanded} accordingly.}
		\begin{equation}
			\label{eq:Psi_Z_def}
			\ket{\Psi_Z}_{Aa} = \frac{1}{\sqrt 2}  \left(\ket{0}_A \ket{\psi_0}_a +\ket{1}_A \ket{\psi_1}_a  \right),
		\end{equation}
		and then measured her ancillary system $A$ in the computational basis {$\{\ket{0}_A,\ket{1}_A\}$}. The objective of our security proof is to estimate the phase-error rate $e_{\rm ph}$, which is defined as the error rate that Alice and Bob would have observed if, in these rounds, Alice had instead measured her ancillary system $A$ in the Hadamard basis $\{\ket{+}_A,\ket{-}_A\}$, where $\ket{\pm}_A = \frac{1}{\sqrt 2}(\ket{0}_A \pm \ket{1}_A)$, and Bob had instead measured $\mathcal{X}$. The asymptotic secret-key rate of the protocol can then be expressed as\footnote{Note that this formula assumes an efficient protocol in which the probability that Alice sends one of the states in $\{\ket{\psi_0},\ket{\psi_1}\}$ and the probability that Bob chooses the POVM $\mathcal{Z}$ both tend to one, which is optimal in the asymptotic regime we are considering.}
		\begin{equation}
			\label{eq:skr_PM}
			R = Y_Z \big[1-h(e_{\rm ph}) - f h(e_Z)\big],
		\end{equation}
		where $Y_Z$ is the observed rate at which Bob obtains a successful detection event conditioned on Alice preparing a state in $\{\ket{\psi_0}_a, \ket{\psi_1}_a\}$ and Bob measuring $\mathcal{Z}$; $e_Z$ is the observed error rate associated to these events; $h(x)=-x \log_2 x - (1-x) \log_2(1-x)$ is the binary entropy function; and $f$ is the error correction inefficiency. The phase-error rate $e_{\rm ph}$ can be expressed as
		\begin{equation}
			\label{eq:e_ph_def}
			e_{\rm ph} = \frac{Y_{Z \wedge \textrm{ph}}}{Y_Z},
		\end{equation}
		where $Y_{Z \wedge \textrm{ph}}$ denotes the probability that Bob obtains a successful detection event \textit{and} Alice and Bob obtain a phase error conditioned on Alice preparing the entangled state in \cref{eq:Psi_Z_def}. Note that, while $Y_Z$ is directly observed in the actual protocol, $Y_{Z \wedge \textrm{ph}}$ is not and must be indirectly estimated. We can write this latter term as
		\begin{equation}
			\label{eq:Y_ph_def_PM}
			\begin{aligned}
				Y_{Z \wedge \textrm{ph}} &= \bra{\Psi_Z} \big(\ketbra{+}_A \otimes \hat{E}_{1_X} + \ketbra{-}_A \otimes \hat{E}_{0_X} \big) \ket{\Psi_Z}_{Aa} \\
				&= \bra{\Psi_Z} \big(\ketbra{+}_A \otimes \hat{M}_{1_X}^\dagger \hat{M}_{1_X} + \ketbra{-}_A \otimes \hat{M}_{0_X}^\dagger \hat{M}_{0_X}  \big) \ket{\Psi_Z}_{Aa} \\
				&= \frac{1}{4} \bigg[\ev{\hat{M}_{1_X}^\dagger \hat{M}_{1_X}} {\psi_0}_a + \mel{\psi_0}{\hat{M}_{1_X}^\dagger \hat{M}_{1_X}} {\psi_1}_a + \mel{\psi_1}{\hat{M}_{1_X}^\dagger \hat{M}_{1_X}} {\psi_0}_a + \ev{\hat{M}_{1_X}^\dagger \hat{M}_{1_X}}{\psi_1}_a  \\ &+  \ev{\hat{M}_{0_X}^\dagger \hat{M}_{0_X}} {\psi_0}_a - \mel{\psi_0}{\hat{M}_{0_X}^\dagger \hat{M}_{0_X}} {\psi_1}_a - \mel{\psi_1}{\hat{M}_{0_X}^\dagger \hat{M}_{0_X}} {\psi_0}_a + \ev{\hat{M}_{0_X}^\dagger \hat{M}_{0_X}}{\psi_1}_a  \bigg].
			\end{aligned}
		\end{equation}
		
		In the second equality of \cref{eq:Y_ph_def_PM}, we have used the fact that, since $\hat E_\gamma \geq 0$, it can be decomposed as $\hat E_\gamma = \hat M_{\gamma}^\dagger \hat M_{\gamma}$ for some $\hat M_{\gamma}$; indeed, one such decomposition can be obtained simply by setting $\hat M_{\gamma} = \hat M_{\gamma}^\dagger = \sqrt{\hat E_\gamma}$. Also, in the third equality, we have substituted the definition of $\ket{\Psi_Z}_{Aa}$ in \cref{eq:Psi_Z_def} and applied some simple algebra.
		
		We now show how to obtain an upper bound on \cref{eq:Y_ph_def_PM} using SDP \cite{Boyd2004}. For clarity, we first consider the case of full state characterization, followed by the case of partial state characterization.

		\subsubsection*{SDP with full state characterization}
		
		In the case of full state characterization, the states emitted by Alice, $\{\ket{\psi_j}_a\}_j$, are completely known. Consequently, the inner products between any two of these states, $\braket{\psi_{j'}}{\psi_j}_a$, can be determined precisely. This allows us to express an upper bound on \cref{eq:Y_ph_def_PM} as the following optimization problem        %
		\begin{equation}
			\label{eq:sdp}
			\begin{aligned}
				&\max \textrm{ } Y_{Z\wedge\textrm{ph}}  \\
				&\textrm{s.t. } \ev{\hat M_{\gamma}^\dagger M_{\gamma}} {\psi_j}_a =  Y_{j}^{\gamma} \quad \forall j, \gamma,\\
				&\sum_\gamma  \mel{\psi_{j'}}{\hat M_{\gamma}^\dagger M_{\gamma}}{\psi_j}_a = \braket{\psi_{j'}}{\psi_j}_a \quad \forall j,j',
			\end{aligned}
		\end{equation}
		where $\gamma \in \{0_X,1_X,f\}$. Note that the first type of constraints in \cref{eq:sdp} holds due to \cref{eq:Y_j_gamma}, and the second type of constraints holds since (see \cref{eq:E_gamma_sum_less_I})
		\begin{equation}
			\label{eq:completeness_eq}
			{\sum_{\gamma \in \{0_X,1_X,f\}} \hat M_{\gamma}^\dagger M_{\gamma} = \sum_{\gamma \in \{0_X,1_X,f\}} \hat E_\gamma = \mathbb{I}.}
		\end{equation}
		
		Importantly, it is straightforward to show that \cref{eq:sdp} is an SDP. Let us define $G$ as the Gram matrix of the vector set $\{\hat M_{\gamma} \ket{\psi_j}_a\}_{j,\gamma}$, i.e., $G$ is a matrix whose elements are given by all the inner products between the vectors in this set. {Clearly, the objective function in \cref{eq:Y_ph_def_PM} and all the constraints in \cref{eq:sdp} are linear with respect to the elements of $G$. Moreover, since $G$ is a Gram matrix, it is necessarily positive semidefinite. Therefore, \cref{eq:sdp} satisfies the criteria of an SDP. That is, it has a linear objective function, linear constraints, and a positive semidefinite variable $G$. Consequently, we can obtain a rigorous upper bound on $Y_{Z\wedge\textrm{ph}}$ by solving its dual problem \cite{primaatmajaVersatileSecurity2019}.}

		\subsubsection*{SDP with partial state characterization}
		
		In this case, the states $\{\ket{\psi_j}_a\}_j$ are only partially characterized, that is, we only know that 
		\begin{equation}
			\label{eq:assumption}
			\abs{\braket{\phi_j}{\psi_j}_a}^2 \geq 1-\epsilon_j,
		\end{equation}
		for some $0<\epsilon_j<1$ and some known states $\{\ket{\phi_j}_a\}_j$. As shown in \cref{app:proof_emitted_states}, \cref{eq:assumption} implies that Alice's emitted states can be assumed to have the form 
		\begin{equation}
			\label{eq:psi_j_trick}
			\ket{\psi_j}_a = \sqrt{1-\epsilon_j} \ket{\phi_j}_a + \sqrt{\epsilon_j} \ket*{\phi_j^\perp}_a,
		\end{equation}
		where $\ket*{\phi_j^\perp}_a$ is a state orthogonal to $\ket*{\phi_j}_a$, i.e., it satisfies $\braket*{\phi_j^\perp}{\phi_j}_a = 0$. Using this, we define our optimization problem as
		\begin{equation}
			\label{eq:sdp_2}
			\begin{aligned}
				&\max \textrm{ } Y_{Z\wedge\textrm{ph}}  \\
				&\textrm{s.t. } (1-\epsilon_j) \mel*{\phi_j}{\hat{M}_{\gamma}^\dagger \hat{M}_{\gamma}}{\phi_j}_a + \sqrt{\epsilon_j(1-\epsilon_j)} \mel*{\phi_j}{\hat{M}_{\gamma}^\dagger \hat{M}_{\gamma}}{\phi_j^\perp}_a \\
				&\quad\quad+ \sqrt{\epsilon_j(1-\epsilon_j)} \mel*{\phi_j^\perp}{\hat{M}_{\gamma}^\dagger \hat{M}_{\gamma}}{\phi_j}_a + \epsilon_j \mel*{\phi_j^\perp}{\hat{M}_{\gamma}^\dagger \hat{M}_{\gamma}}{\phi_j^\perp}_a =  Y_{j}^{\gamma} \quad \forall j, \gamma,\\
				&\sum_\gamma  \mel{\phi_{j'}}{\hat M_{\gamma}^\dagger \hat M_{\gamma}}{\phi_j}_a = \braket{\phi_{j'}}{\phi_j}_a \quad \forall j,j', \\
				&\sum_\gamma  \mel*{\phi_{j}^\perp}{\hat M_{\gamma}^\dagger \hat M_{\gamma}}{\phi_j}_a = 0 \quad \forall j, \\
				&\sum_\gamma  \mel*{\phi_{j}^\perp}{\hat M_{\gamma}^\dagger \hat M_{\gamma}}{\phi_j^\perp}_a = 1 \quad \forall j,
			\end{aligned}
		\end{equation}
		where $\gamma \in \{0_X,1_X,f\}$ and whose objective function is given by
		\begin{equation}
			\label{eq:Y_ph_def_PM_expanded}
			\begin{aligned}
				Y_{Z \wedge \textrm{ph}} &= \frac{1}{4} \bigg[ (1-\epsilon_0) \ev*{\hat{M}_{1_X}^\dagger \hat{M}_{1_X}}{\phi_0}_a + \sqrt{\epsilon_0(1-\epsilon_0)} \mel*{\phi_0}{\hat{M}_{1_X}^\dagger \hat{M}_{1_X}}{\phi_0^\perp}_a \\
				&+ \sqrt{\epsilon_0(1-\epsilon_0)} \mel*{\phi_0^\perp}{\hat{M}_{1_X}^\dagger \hat{M}_{1_X}}{\phi_0}_a + \epsilon_0 \ev*{\hat{M}_{1_X}^\dagger \hat{M}_{1_X}}{\phi_0^\perp}_a \\
				&+ \sqrt{(1-\epsilon_0)(1-\epsilon_1)} \mel*{\phi_0}{\hat{M}_{1_X}^\dagger \hat{M}_{1_X}}{\phi_1}_a + \sqrt{(1-\epsilon_0)\epsilon_1} \mel*{\phi_0}{\hat{M}_{1_X}^\dagger \hat{M}_{1_X}}{\phi_1^\perp}_a \\
				&+ \sqrt{\epsilon_0(1-\epsilon_1)} \mel*{\phi_0^\perp}{\hat{M}_{1_X}^\dagger \hat{M}_{1_X}}{\phi_1}_a + \sqrt{\epsilon_0\epsilon_1} \mel*{\phi_0^\perp}{\hat{M}_{1_X}^\dagger \hat{M}_{1_X}}{\phi_1^\perp}_a \\
				&+ \sqrt{(1-\epsilon_1)(1-\epsilon_0)} \mel*{\phi_1}{\hat{M}_{1_X}^\dagger \hat{M}_{1_X}}{\phi_0}_a + \sqrt{(1-\epsilon_1)\epsilon_0} \mel*{\phi_1}{\hat{M}_{1_X}^\dagger \hat{M}_{1_X}}{\phi_0^\perp}_a \\
				&+ \sqrt{\epsilon_1(1-\epsilon_0)} \mel*{\phi_1^\perp}{\hat{M}_{1_X}^\dagger \hat{M}_{1_X}}{\phi_0}_a + \sqrt{\epsilon_1\epsilon_0} \mel*{\phi_1^\perp}{\hat{M}_{1_X}^\dagger \hat{M}_{1_X}}{\phi_0^\perp}_a \\
				&+ (1-\epsilon_1) \ev*{\hat{M}_{1_X}^\dagger \hat{M}_{1_X}}{\phi_1}_a + \sqrt{\epsilon_1(1-\epsilon_1)} \mel*{\phi_1}{\hat{M}_{1_X}^\dagger \hat{M}_{1_X}}{\phi_1^\perp}_a \\
				&+ \sqrt{\epsilon_1(1-\epsilon_1)} \mel*{\phi_1^\perp}{\hat{M}_{1_X}^\dagger \hat{M}_{1_X}}{\phi_1}_a + \epsilon_1 \ev*{\hat{M}_{1_X}^\dagger \hat{M}_{1_X}}{\phi_1^\perp}_a \\
				&+ (1-\epsilon_0) \ev*{\hat{M}_{0_X}^\dagger \hat{M}_{0_X}}{\phi_0}_a + \sqrt{\epsilon_0(1-\epsilon_0)} \mel*{\phi_0}{\hat{M}_{0_X}^\dagger \hat{M}_{0_X}}{\phi_0^\perp}_a \\
				&+ \sqrt{\epsilon_0(1-\epsilon_0)} \mel*{\phi_0^\perp}{\hat{M}_{0_X}^\dagger \hat{M}_{0_X}}{\phi_0}_a + \epsilon_0 \ev*{\hat{M}_{0_X}^\dagger \hat{M}_{0_X}}{\phi_0^\perp}_a \\
				&- \sqrt{(1-\epsilon_0)(1-\epsilon_1)} \mel*{\phi_0}{\hat{M}_{0_X}^\dagger \hat{M}_{0_X}}{\phi_1}_a - \sqrt{(1-\epsilon_0)\epsilon_1} \mel*{\phi_0}{\hat{M}_{0_X}^\dagger \hat{M}_{0_X}}{\phi_1^\perp}_a \\
				&- \sqrt{\epsilon_0(1-\epsilon_1)} \mel*{\phi_0^\perp}{\hat{M}_{0_X}^\dagger \hat{M}_{0_X}}{\phi_1}_a - \sqrt{\epsilon_0\epsilon_1} \mel*{\phi_0^\perp}{\hat{M}_{0_X}^\dagger \hat{M}_{0_X}}{\phi_1^\perp}_a \\
				&- \sqrt{(1-\epsilon_1)(1-\epsilon_0)} \mel*{\phi_1}{\hat{M}_{0_X}^\dagger \hat{M}_{0_X}}{\phi_0}_a - \sqrt{(1-\epsilon_1)\epsilon_0} \mel*{\phi_1}{\hat{M}_{0_X}^\dagger \hat{M}_{0_X}}{\phi_0^\perp}_a \\
				&- \sqrt{\epsilon_1(1-\epsilon_0)} \mel*{\phi_1^\perp}{\hat{M}_{0_X}^\dagger \hat{M}_{0_X}}{\phi_0}_a - \sqrt{\epsilon_1\epsilon_0} \mel*{\phi_1^\perp}{\hat{M}_{0_X}^\dagger \hat{M}_{0_X}}{\phi_0^\perp}_a \\
				&+ (1-\epsilon_1) \ev*{\hat{M}_{0_X}^\dagger \hat{M}_{0_X}}{\phi_1}_a + \sqrt{\epsilon_1(1-\epsilon_1)} \mel*{\phi_1}{\hat{M}_{0_X}^\dagger \hat{M}_{0_X}}{\phi_1^\perp}_a \\
				&+ \sqrt{\epsilon_1(1-\epsilon_1)} \mel*{\phi_1^\perp}{\hat{M}_{0_X}^\dagger \hat{M}_{0_X}}{\phi_1}_a + \epsilon_1 \ev*{\hat{M}_{0_X}^\dagger \hat{M}_{0_X}}{\phi_1^\perp}_a \bigg].
			\end{aligned}
		\end{equation} 
		To obtain the first type of constraints in \cref{eq:sdp_2}, we have substituted \cref{eq:psi_j_trick} into \cref{eq:Y_j_gamma}, and the second, third and fourth types of constraints hold due to \cref{eq:completeness_eq}. Also, to obtain the objective function in \cref{eq:Y_ph_def_PM_expanded}, we have substituted \cref{eq:psi_j_trick} into \cref{eq:Y_ph_def_PM}.
		
		Again, it is straightforward to show that \cref{eq:sdp_2} is an SDP. Let us define $G$ as the Gram matrix of the union of the vector sets $\{\hat M_\gamma \ket{\phi_j}_a\}_{j,\gamma}$ and $\{\hat M_\gamma \ket*{\phi_j^\perp}_a\}_{j,\gamma}$. Clearly, all the constraints in \cref{eq:sdp_2} are linear with respect to the elements of $G$. Similarly, the objective function in \cref{eq:Y_ph_def_PM_expanded} is a linear function of elements of $G$. Moreover, since $G$ is a Gram matrix, it is necessarily positive semidefinite. Hence, \cref{eq:sdp_2} satisfies all requirements of an SDP, and we can obtain a rigorous upper bound on $Y_{Z\wedge\textrm{ph}}$ by solving its dual problem. Finally, by substituting this bound in \cref{eq:e_ph_def}, we obtain a rigorous upper bound on the phase-error rate.

		\section{Numerical results}
		\label{sec:results}
		
		In this section, for illustration purposes, we evaluate the secret-key rate achievable using our security proof for some selected protocols when only a partial characterization of the emitted states is available. We also compare our results with those obtained using state-of-the-art analytical techniques that can handle this scenario.
		
		In our simulations, we consider, for all protocols, a standard optical fiber channel with attenuation coefficient $\alpha_{\rm dB} = \SI{0.2}{\decibel\per\kilo\metre}$, and single-photon detectors of efficiency $\eta_{\rm det} = 0.73$ and dark count probability $p_d = 10^{-8}$ \cite{pittaluga600kmRepeaterlike2021}. The error correction inefficiency is set to $f = 1.16$ (see the secret-key rate formula in \cref{eq:skr_PM}). For the detailed channel models, see Ref.~\cite[Supplementary Equations 3]{curras-lorenzoSecurityFramework2023} for \cref{fig:bb84_and_3state,fig:graph_mdi_overall}, and Ref.~\cite[Appendix C]{navarreteImprovedFiniteKey2022} for \cref{fig:graph_decoy}.

		\subsection{BB84-type protocols}
		
		Here, we apply our analysis to a BB84 protocol {with an imperfect source, for both its standard four-state version and its alternative three-state version \cite{boileauUnconditionalSecurity2005,tamakiLosstolerantQuantum2014}}. We consider that Alice's states $\{\ket{\psi_j}_{a}\}_j$ are $\epsilon$-close in fidelity to known qubit states $\{\ket{\phi_j}_{a}\}_j$ in the XZ plane, i.e.,
		\begin{equation}
			\ket{\phi_{j}}_a = \cos (\varphi_j/2) \ket{0}_a +\sin (\varphi_j/2) \ket{1}_a ,
		\end{equation}
		where $j \in \{0,1,+,-\}$ for {the four-state scenario} and  $j \in \{0,1,+\}$ for the {three-state scenario}. We model SPFs as small deviations of the encoding angles $\varphi_j$ from their ideal values $\hat \varphi_{0} = 0$, $\hat \varphi_{+} = \pi/2$, $\hat \varphi_{1} = \pi$, and $\hat \varphi_{-} = 3\pi/2$, such that
		\begin{equation}
			\varphi_{j} = (1+\delta/\pi) \hat \varphi_{j},
		\end{equation}
		where $\delta \in [0,\pi)$ quantifies the magnitude of the SPFs.
		
		\Cref{fig:graph_bb84,fig:graph_3state} compare the asymptotic secret-key rates obtained using our analysis and the state-of-the-art analytical techniques from \cite{curras-lorenzoSecurityFramework2023}, which provide a complete finite-key security proof against coherent attacks. For the four-state BB84 protocol (\cref{fig:graph_bb84}), our numerical approach—which obtains an optimal bound on the phase-error rate given the available information—offers almost no improvement. This suggests that the security proof in \cite{curras-lorenzoSecurityFramework2023} is remarkably tight, at least in the asymptotic regime. However, for the three-state protocol (\cref{fig:graph_3state}), our method provides a larger advantage. In fact, interestingly, our analysis yields similar key rates for both protocols, contradicting the findings in \cite{pereiraModifiedBB842023,curras-lorenzoSecurityFramework2023}, which suggested that{, when $\epsilon > 0$, one could obtain appreciably higher key rates by sending four states}. This suggests that, while the analyses in \cite{pereiraModifiedBB842023,curras-lorenzoSecurityFramework2023} are remarkably tight for the {four-state scenario}, they seem to be {slightly} less tight for the {three-state scenario}. Consequently, the previously reported key rate differences between these protocols in the presence of side channels can likely be attributed to this gap, rather than inherent limitations of the three-state protocol.

		\begin{figure}[h]%
			\centering
			\subfloat[\centering]{{\includegraphics[width=8.55cm]{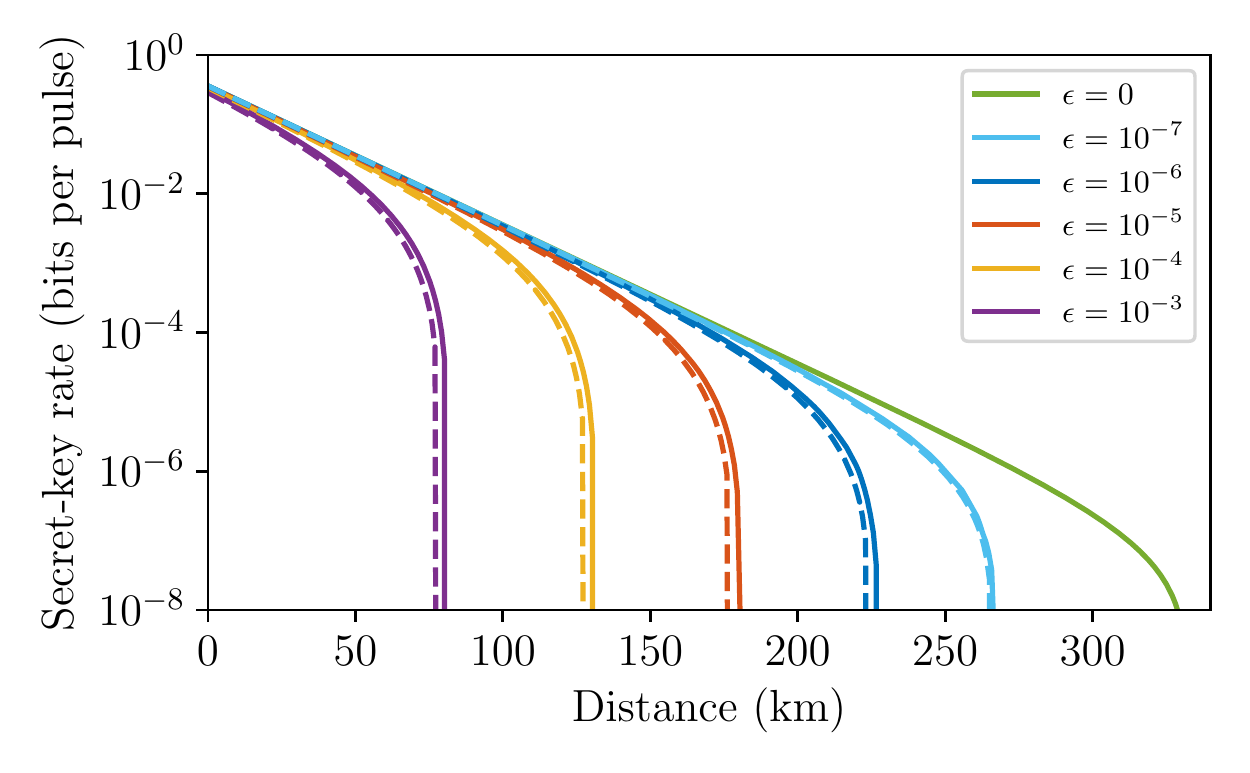}}\label{fig:graph_bb84}}%
			\qquad
			\subfloat[\centering]{{\includegraphics[width=8.55cm]{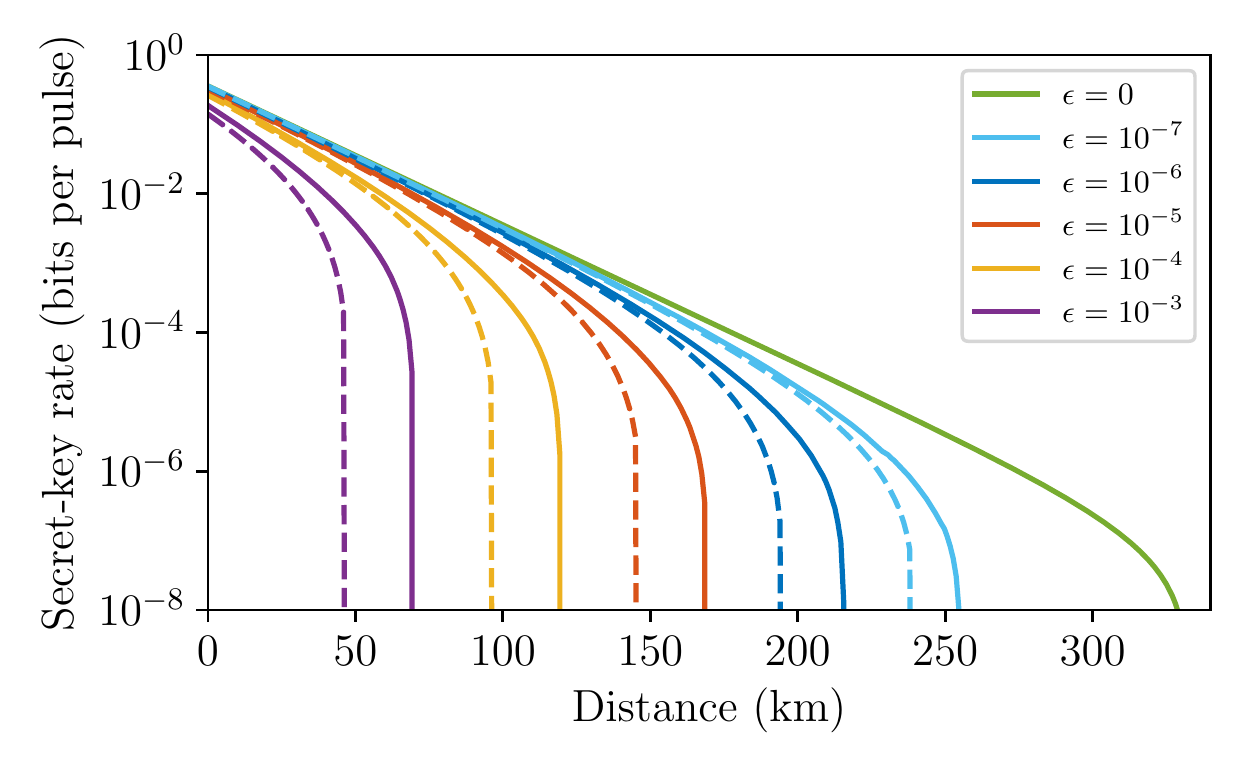}}\label{fig:graph_3state}}%
			\caption{Asymptotic secret-key rates for {the BB84 protocol with an imperfect source in its (a) standard four-state version and (b) alternative three-state version \cite{boileauUnconditionalSecurity2005,tamakiLosstolerantQuantum2014}} using the numerical analysis presented in this work (solid lines) compared with the analytical results in \cite{curras-lorenzoSecurityFramework2023} (dashed lines). We consider $\delta = 0.063$ \cite{honjoDifferentialphaseshiftQuantum2004,xuExperimentalQuantum2015}   and several values of $\epsilon$, {which correspond to the magnitudes of the characterized SPFs and the uncharacterized imperfections, respectively.}}%
			\label{fig:bb84_and_3state}%
		\end{figure}

		\subsection{Coherent-light-based MDI protocol}
		
		Next, we apply our analysis to a coherent-light-based MDI-type protocol introduced in \cite{navarretePracticalQuantum2021}. We have chosen this scheme due to its non-qubit encoding, allowing us to compare the performance of our techniques with those of existing analytical methods for such scenarios.
		
		Ideally, in this scheme, Alice and Bob  emit states from the set $\{\ket*{\sqrt \mu}, \ket*{-\sqrt \mu}, \ket{\rm vac}\}$, as illustrated in \cref{fig:coherentMDI}. The key generation states $\ket{\sqrt\mu}$ and $\ket{-\sqrt\mu}$ are coherent states of fixed intensity $\mu$ and opposite phases, associated with the bit values 0 and 1, respectively. The state $\ket{\rm vac}$ is a vacuum state that can be generated by turning off the laser and contributes only to the phase-error rate estimation. {As before, we assume that the actual states emitted by each of Alice and Bob are only $\epsilon$-close in fidelity to these ideal states.}

		\begin{figure}[h]
			\centering

			\includegraphics[width=0.45\columnwidth]{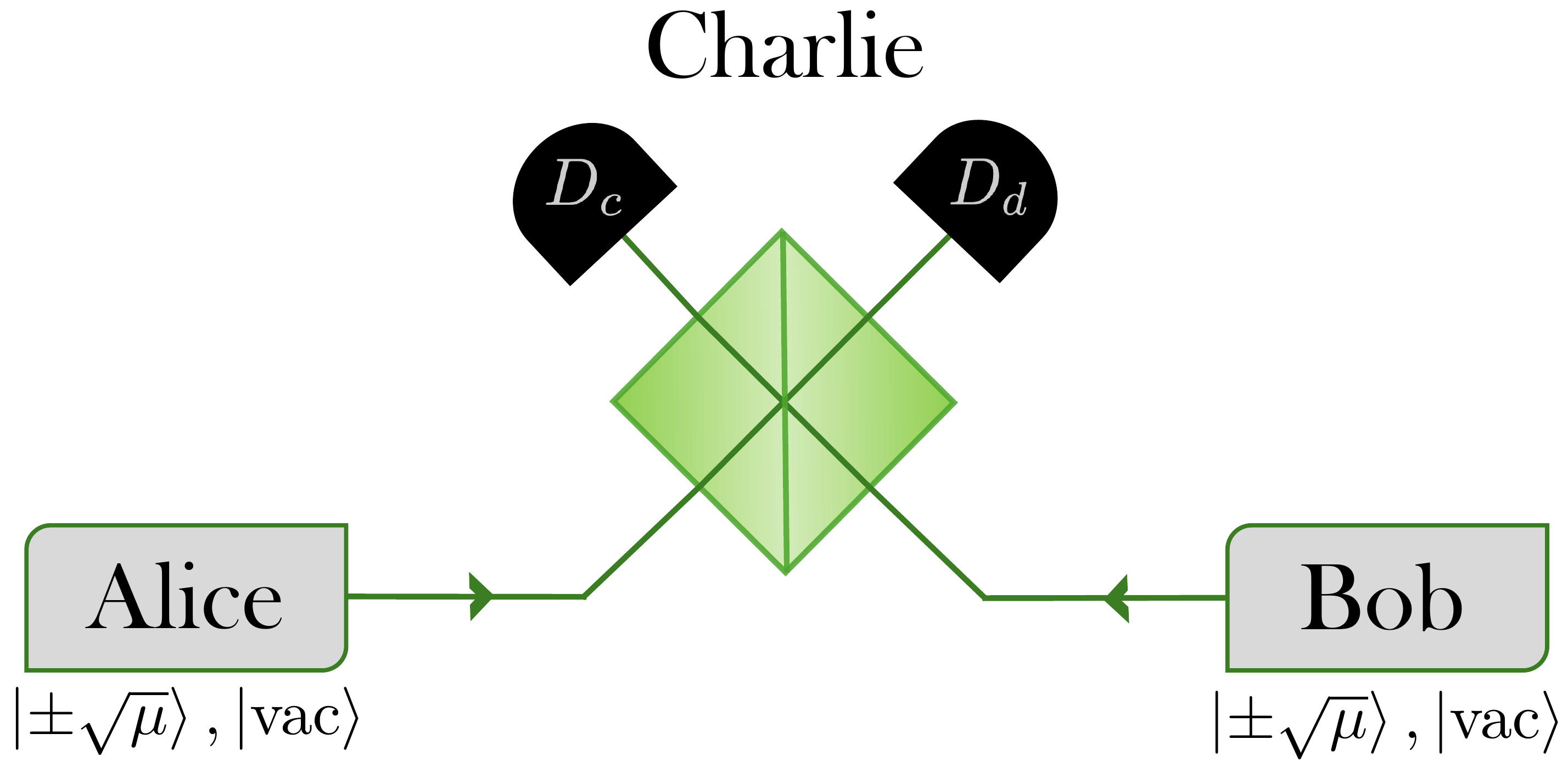}

			\caption{Illustration of the coherent-light-based MDI-type setup introduced in \cite{navarretePracticalQuantum2021}. In every round, each of Alice and Bob (ideally) prepares a state in the set $\{\ket*{\sqrt \mu}, \ket*{-\sqrt \mu}, \ket{\rm vac}\}$ and sends it to the untrusted middle node Charlie through a quantum channel. Charlie measures the incoming signals by interfering them with a 50:50 beamsplitter, followed by two threshold single-photon detectors $D_c$ and $D_d$, which are associated with constructive and destructive interference, respectively; and announces the outcomes.}
			\label{fig:coherentMDI}
		\end{figure}

		\Cref{fig:graph_coh_mdi} shows the secret-key rates achievable using our analysis for several values of $\xi \coloneqq 1-(1-\epsilon)^2 \approx 2 \epsilon$, and compares them with those achievable with the analytical techniques in \cite{curras-lorenzoSecurityFramework2023}. For this protocol, our analysis provides a significant performance improvement in both the achievable secret-key rate at any particular distance and the maximum distance at which a non-zero key rate can be obtained. 
		The reason for this remarkable improvement seems to be that our analysis provides significantly tighter bounds on the phase-error rate, particularly for higher values of the laser intensity $\mu$. In fact, our simulation results suggest that the optimal value of $\mu$ for our analysis is around $1.7$ times higher than that of \cite{curras-lorenzoSecurityFramework2023} for most distances. We conjecture that this is because the set of states $\{\ket*{\sqrt \mu}, \ket*{-\sqrt \mu}, \ket{\rm vac}\}$ approach qubit states as $\mu \to 0$; specifically, the component of $\ket{\rm vac}$ that lies outside the qubit space spanned by $\{\ket*{\sqrt \mu}, \ket*{-\sqrt \mu}\}$ vanishes as $\mu \to 0$.  Since the security proof in \cite{curras-lorenzoSecurityFramework2023} is ideally suited for qubit encoding scenarios, it provides tighter results for low $\mu$ but becomes looser as $\mu$ increases. Consequently, its optimal values of $\mu$ are lower than those of our numerical analysis, resulting in a worse overall performance.

		\begin{figure}[h]%
			\centering
			\includegraphics[width=8.55cm]{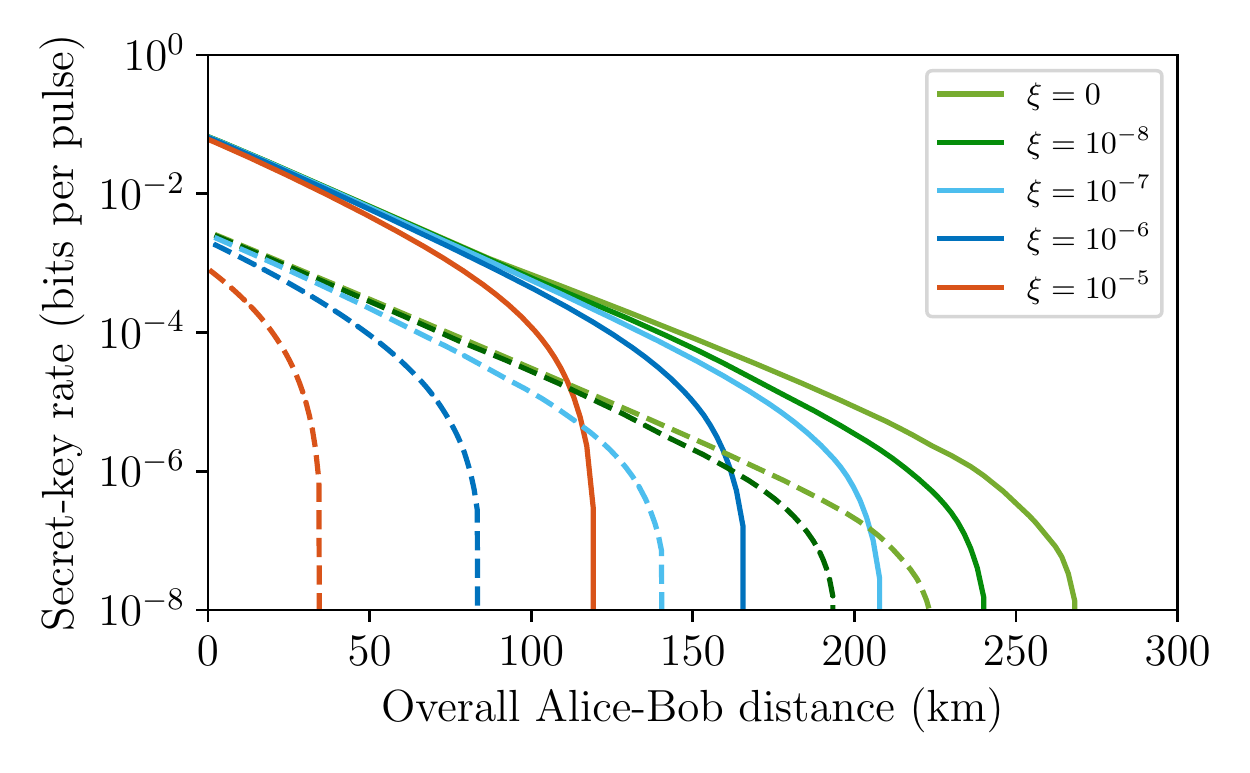}
			\caption{Asymptotic secret-key rates for the coherent-light-based MDI protocol \cite{navarretePracticalQuantum2021} illustrated in \cref{fig:coherentMDI} using the numerical analysis presented in this work (solid lines) compared with the analytical results in \cite{curras-lorenzoSecurityFramework2023} (dashed lines). We {consider several values of $\xi \coloneqq 1-(1-\epsilon)^2 \approx 2 \epsilon$ and} optimize over the value of the coherent-light intensity $\mu$ for each distance value.}%
			\label{fig:graph_coh_mdi}%
			\label{fig:graph_mdi_overall}
		\end{figure}

		\subsection{Trojan-horse attack against a decoy-state protocol}
		\label{subsec:decoy_THA_results}
		
		Our numerical framework extends naturally to decoy-state QKD protocols, as detailed in \cref{app:decoy_state}. In fact, for such protocols, our analysis can incorporate not only imperfections in the bit-and-basis encoding, but also vulnerabilities in intensity modulation, such as information leakage about Alice's intensity setting choices. 
		
		As an example to demonstrate the practical utility of our analysis, we simulate the achievable secret-key rate in a scenario where Eve launches a Trojan-horse attack (THA) targeting both the intensity modulator and the bit-and-basis encoder \cite{vakhitovLargePulse2001,gisinTrojanhorseAttacks2006,lucamariniPracticalSecurity2015,navarreteImprovedFiniteKey2022,sixtoQuantumKey2024}. A key advantage of our analysis is its characterization simplicity: for this scenario, it only requires knowledge of a single parameter, an upper bound $I_{\rm max}$ on the intensity (average photon number) of the back-reflected light. That is, no detailed characterization of the quantum state of the injected or back-reflected light is needed.
		
		{The parameter $I_{\rm max}$ could be established through the following process: (1) install an optical fuse or power limiter to restrict Eve's per-round injected light intensity to some maximum value $I_{\rm in}^{\rm U}$ (or alternatively, establish Eve's maximum input intensity by considering the laser induced damage threshold \cite{lucamariniPracticalSecurity2015}); (2) implement optical isolation in the source to ensure that the back-reflected light's power is limited to a fraction $\eta^{\rm U}$ of Eve's injected light; and (3) calculate $I_{\rm max} = \eta^{\rm U}I_{\rm in}^{\rm U}$. Note that, in principle, any desired value of $I_{\rm max}$ could be achieved by simply increasing the source's isolation to reduce the value of $\eta^{\rm U}$.} \Cref{fig:graph_decoy} shows how the achievable secret-key rate of decoy-state QKD varies with $I_{\rm max}$. For more information on the application of our analysis to this scenario, see \cref{appsec:THA_analysis}.

		\begin{figure}[h]%
			\centering
			{
				\includegraphics[width=8.55cm]{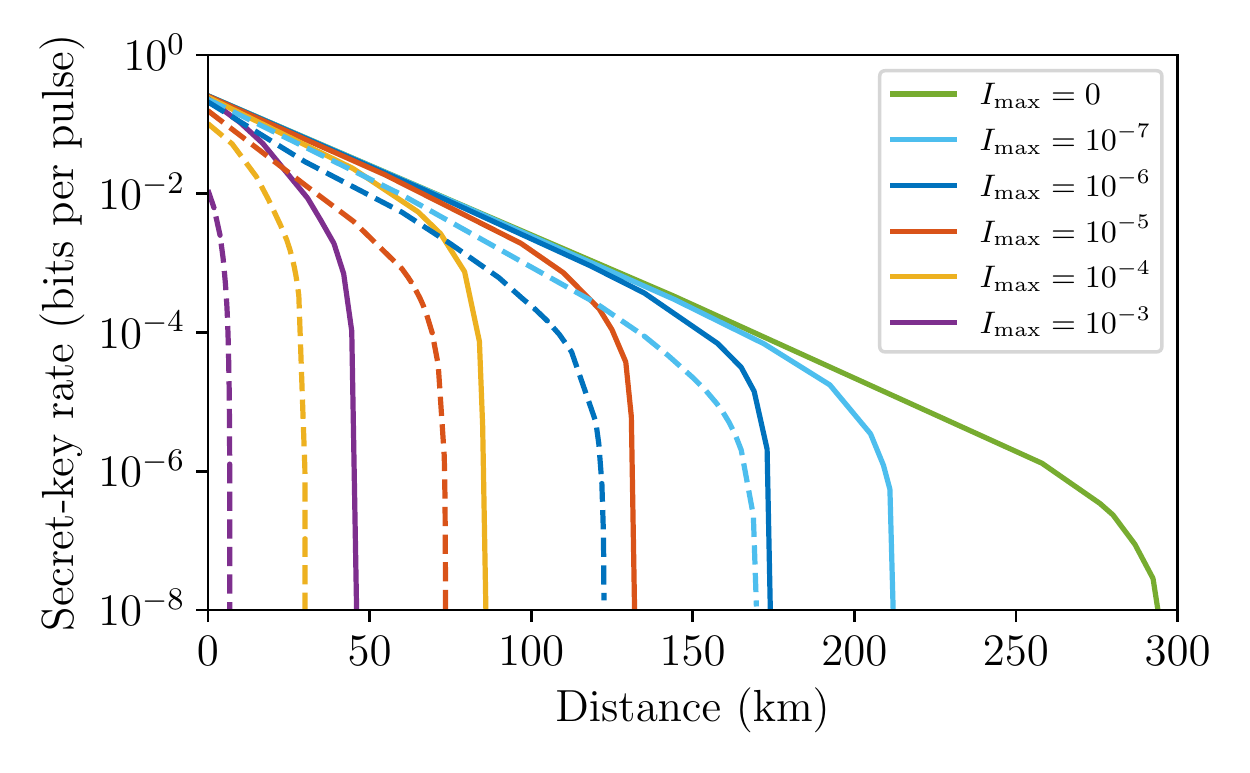}
			}%
			\caption{Asymptotic secret-key rates as a function of the maximum intensity ($I_{\rm max}$) of the back-reflected light for a decoy-state BB84 protocol under a Trojan-horse attack (THA) when using our analysis {(solid lines), compared with the analytical results in \cite{navarreteImprovedFiniteKey2022} (dashed lines)}. We assume that Alice uses three intensity settings $\{\mu_0,\mu_1,\mu_2\}$, where for simplicity we set {$\mu_1 = 0.02$},  $\mu_2 =0$, and optimize over the value of $\mu_0$ for each distance value. 
			}%
			\label{fig:graph_decoy}
		\end{figure}

		\section{Discussion and Conclusion}
		\label{sec:discussion_conclusion}
		
		In this paper, we have introduced a novel numerical security proof approach based on SDP that can evaluate the asymptotic secret-key rate of prepare-and-measure and measurement-device-independent QKD protocols when only partial information about the emitted states is available. This represents a significant advancement over previous numerical techniques \sdpcitations{}, which typically require perfect state characterization, a condition that is very difficult to meet in practical implementations. Our methodology provides an optimal (up to numerical precision) upper bound on the phase-error rate given the available information about the emitted states, and consequently, an optimal lower bound on the asymptotic secret-key rate within the entropic uncertainty relation \cite{tomamichelTightFinitekey2012} and phase-error correction \cite{koashiSimpleSecurity2009} frameworks. We have compared the asymptotic secret-key rate obtainable using our results with those of existing analytical techniques \cite{curras-lorenzoSecurityFramework2023} and found that our approach yields higher key rates, particularly for protocols with non-qubit encoding spaces. This demonstrates its potential to improve the performance of QKD systems in realistic scenarios.
		
		Our current work focuses on deriving asymptotic bounds on the secret-key rate, in line with many other recent numerical security proofs \sdpcitationsasymp{}. In contrast, the analytical proof in \cite{curras-lorenzoSecurityFramework2023} is directly applicable in the finite-size regime. To transform our result into a finite-size security proof valid against general attacks, one could apply known concentration bounds (see, e.g., \cite{georgeNumericalCalculations2021,bunandarNumericalFinitekey2020}) and then use established techniques to lift the security proof from collective attacks to general attacks, such as the generalized entropy accumulation theorem (GEAT) \cite{metgerSecurityQuantum2023} or the postselection technique (PST) \cite{christandlPostselectionTechnique2009}. However, applying the GEAT introduces a sequential condition that reduces the protocol's repetition rate, which is not required in \cite{curras-lorenzoSecurityFramework2023}. On the other hand, employing the PST can lead to a considerable reduction in the achievable finite-size secret-key rate, although recent improvements \cite{naharPostselectionTechnique2024} may partially mitigate this issue.  Consequently, it remains an open question whether or not our results could outperform those of the analytical security proof in \cite{curras-lorenzoSecurityFramework2023} not only in the asymptotic regime but also for practical finite-size block lengths and considering general coherent attacks.
		
		Despite this open question, our work represents a significant step forward in the development of practical and secure QKD systems. The flexibility of our numerical approach allows for the rapid evaluation of the asymptotic key rates for a wide range of protocols, providing valuable insights into their potential performance under partial source characterization and guiding the design of future QKD implementations. Moreover, our technique can be easily adapted to incorporate additional constraints or information about the emitted states as they become available, further enhancing its applicability to real-world scenarios.
		
		In conclusion, our novel SDP-based numerical security proof technique addresses a critical gap in existing numerical methods by enabling the evaluation of asymptotic secret-key rates for QKD protocols with partial state characterization. By doing so, it not only advances the state-of-the-art in QKD security proofs but also paves the way for more practical and secure quantum communication systems. 

		\section*{ACKNOWLEDGEMENTS}
		
		{We thank Mateus Araújo, Kiyoshi Tamaki and Go Kato for valuable discussions.  This work was supported by the Galician Regional Government (consolidation of Research Units: AtlantTIC), the Spanish Ministry of
			Economy and Competitiveness (MINECO), the Fondo
			Europeo de Desarrollo Regional (FEDER) through
			the grant No.~PID2020-118178RB-C21, MICIN with
			funding from the European Union NextGenerationEU
			(PRTR-C17.I1) and the Galician Regional Government
			with own funding through the "Planes Complementarios de I+D+I con las Comunidades Aut\'onomas" in
			Quantum Communication, the European Union’s Horizon Europe Framework Programme under the Marie
			Sk\l{}odowska-Curie Grant No.~101072637 (Project QSI)
			and the project "Quantum Security Networks Partnership" (QSNP, grant agreement No 101114043).}
		
		\appendix

		\section{Application to MDI-type protocols}
		\label{app:MDI}

		{Here, we present our analysis for general MDI-type protocols \cite{loMeasurementDeviceIndependentQuantum2012}, which is quite similar to that for P$\&$M protocols presented in \cref{sec:PM}. For simplicity, here, we consider a discretely-modulated protocol. However, we remark that our analysis can be applied to non-discretely-modulated MDI-type cases, such as decoy-state MDI-type scenarios, by following the approach in \cref{app:decoy_state}.
			
			In a general discretely-modulated MDI-type scenario, for} each round, Alice and Bob each probabilistically selects a state from the sets $\{\ket{\psi_i}_a\}_i$ and $\{\ket{\psi_j}_b\}_j$, respectively, and sends it to the untrusted middle node Charlie through the quantum channel. Charlie then performs a measurement on the received signals and announces an outcome $\gamma$. Note that, in general, the states emitted by Alice and Bob could be different. For simplicity, in this discussion, we will assume that Charlie reports either a successful measurement or a failure, i.e., that $\gamma \in \{\texttt{pass},\texttt{fail}\}$. However, we remark that our analysis can also be applied when Charlie reports more than one successful result. For example, in standard MDI-QKD \cite{loMeasurementDeviceIndependentQuantum2012}, Charlie can report a projection to the singlet state, the triplet state or a failure. 
		
		In MDI-type protocols, Eve may not only control the quantum channel but also the middle node Charlie, and all the announcements that he makes. Therefore, in this scenario, Eve's collective attack can be described as an isometry
		\begin{equation}
			\ket{\psi_i}_a\ket{\psi_j}_b \to \sum_{\gamma}\hat{M}_{\gamma} \ket{\psi_i}_a\ket{\psi_j}_b \ket{\gamma}_E,
		\end{equation}
		where $\{\hat M_\gamma\}$ is a set of Kraus operators such that
		\begin{equation}
			\label{eq:completeness_eq_mdi}
			{\sum_{\gamma \in \{\texttt{pass},\texttt{fail}\}} \hat M_\gamma^\dagger \hat M_\gamma = \mathbb{I}.}
		\end{equation}
		The probability that Eve reports the outcome $\gamma$ when Alice and Bob select the settings $i$ and $j$, respectively, can be expressed as
		\begin{equation}
			\label{eq:Y_ij_gamma}
			\begin{gathered}
				Y_{i,j}^{\gamma} = \bra{\psi_i}\!\!\mel{\psi_j}{\hat M_{\gamma}^\dagger \hat M_{\gamma}}{\psi_i}_a\!\ket{\psi_j}_b. \\      
			\end{gathered}
		\end{equation}
		
		For concreteness, let us assume that Alice and Bob obtain their sifted key from the events in which Alice emits $\ket{\psi_0}_a$ or $\ket{\psi_1}_a$, Bob emits $\ket{\psi_0}_b$ or $\ket{\psi_1}_b$, and Charlie announces a successful detection. Also, let us assume that Alice (Bob) emits $\ket{\psi_0}_a$ and $\ket{\psi_1}_a$ ($\ket{\psi_0}_b$ and $\ket{\psi_1}_b$) equiprobabilistically.  Note that, in these rounds, Alice and Bob could have equivalently generated the entangled states
		\begin{equation}
			\label{eq:Psi_Z_def_mdi}
			\ket{\Psi_Z}_{AaBb} = \frac{1}{\sqrt 2}  \left(\ket{0}_A \ket{\psi_0}_a +\ket{1}_A \ket{\psi_1}_a  \right) \otimes \frac{1}{\sqrt 2}  \left(\ket{0}_B \ket{\psi_0}_b +\ket{1}_B \ket{\psi_1}_b  \right), 
		\end{equation}
		and then measured their respective ancillary systems $A$ and $B$ in the computational basis $\{\ket{0},\ket{1}\}$. The objective of the security proof is to estimate the phase-error rate $e_{\rm ph}$, which is defined as the error rate that Alice and Bob would have observed if, in these key rounds, they had instead measured their respective ancillary systems $A$ and $B$ in the Hadamard basis $\{\ket{+},\ket{-}\}$\footnote{We remark that the appropriate definition of a phase error depends on the MDI-type protocol in consideration. In \cref{eq:Y_ph_def_MDI}, we assume for concreteness that a phase error is defined as an event in which Alice and Bob obtain the same $X$-basis result, i.e. $\ket{+}_A$ and $\ket{+}_B$ or $\ket{-}_A$ and $\ket{-}_B$, which is the appropriate definition for the coherent-light-based MDI-type protocol depicted in \cref{fig:coherentMDI}. However, in some protocols, the appropriate definition of a phase error actually depends on the specific announcement made by Charlie. For example, in the standard MDI-QKD scheme \cite{loMeasurementDeviceIndependentQuantum2012}, a phase-error should be defined as an event in which Alice and Bob obtain the same (different) $X$-basis result when Charlie reports a projection to the triplet (singlet) state; see Ref.~\cite[Table I]{loMeasurementDeviceIndependentQuantum2012}. In such cases, one should define separate sifted key pairs according to the specific announcement made by Charlie, and apply our analysis to compute a bound on the phase-error rate separately for each sifted key pair. When defining phase errors as events in which Alice and Bob obtain a different $X$-basis result, one should define the phase-error probability as $Y_{Z \wedge \textrm{ph}} = \bra{\Psi_Z} \hat M_{\tt pass}^\dagger \big(\ketbra{+}_A \otimes \ketbra{-}_B + \ketbra{-}_A \otimes \ketbra{+}_B \big)  \hat M_{\tt pass} \ket{\Psi_Z}_{AaBb}$, and update the objective functions in \cref{eq:Y_ph_def_MDI,eq:Y_ph_def_MDI_expanded} accordingly.}. The asymptotic secret-key rate formula is the same as in the case of P$\&$M protocols and is given by \cref{eq:skr_PM}, where 
		$Y_Z$ now refers to the observed rate at which Eve announces a successful detection conditioned on Alice preparing a state in $\{\ket{\psi_0}_a, \ket{\psi_1}_a\}$ and Bob preparing a state in $\{\ket{\psi_0}_b, \ket{\psi_1}_b\}$, and $e_Z$ is the observed error rate associated to these events. Similarly, the phase-error rate $e_{\rm ph}$ is given by \cref{eq:e_ph_def},
		where $Y_{Z \wedge \textrm{ph}}$ now denotes the probability that Eve announces a successful detection \textit{and} Alice and Bob obtain a phase error conditioned on Alice and Bob preparing $\ket{\Psi_Z}_{AaBb}$ and measuring systems $A$ and $B$ in the Hadamard basis. 
		For the MDI-type protocol depicted in \cref{fig:coherentMDI}, we can write the term {$Y_{Z \wedge \textrm{ph}}$} as
		\begin{equation}
			\label{eq:Y_ph_def_MDI}
			\begin{aligned}
				Y_{Z \wedge \textrm{ph}} &= \bra{\Psi_Z} \hat M_{\tt pass}^\dagger \big(\ketbra{+}_A \otimes \ketbra{+}_B + \ketbra{-}_A \otimes \ketbra{-}_B \big)  \hat M_{\tt pass} \ket{\Psi_Z}_{AaBb} \\
				&= \frac{1}{8}\Big[\mel{\psi_{00}}{\hat M_{\tt pass}^\dagger \hat M_{\tt pass}}{\psi_{00}}_{ab} + \mel{\psi_{00}}{\hat M_{\tt pass}^\dagger \hat M_{\tt pass}}{\psi_{11}}_{ab} + \mel{\psi_{01}}{\hat M_{\tt pass}^\dagger \hat M_{\tt pass}}{\psi_{01}}_{ab} + \mel{\psi_{01}}{\hat M_{\tt pass}^\dagger \hat M_{\tt pass}}{\psi_{10}}_{ab} \\
				&+ \mel{\psi_{10}}{\hat M_{\tt pass}^\dagger \hat M_{\tt pass}}{\psi_{01}}_{ab} + \mel{\psi_{10}}{\hat M_{\tt pass}^\dagger \hat M_{\tt pass}}{\psi_{10}}_{ab} + \mel{\psi_{11}}{\hat M_{\tt pass}^\dagger \hat M_{\tt pass}}{\psi_{00}}_{ab} + \mel{\psi_{11}}{\hat M_{\tt pass}^\dagger \hat M_{\tt pass}}{\psi_{11}}_{ab}\Big],
			\end{aligned}
		\end{equation}
		{where we have introduced the shorthand notation $\ket{\psi_{ij}}_{ab} \coloneqq \ket{\psi_i}_a \otimes \ket{\psi_j}_b$.}
		
		We now show how to obtain an upper bound on \cref{eq:Y_ph_def_MDI} using SDP. For clarity, we first consider the case of full state characterization, followed by the case of partial state characterization.

		\subsection*{SDP with full state characterization}
		
		In this case, the joint states emitted by Alice and Bob, $\{\ket{\psi_{ij}}_{ab}\}_{ij}$, are characterized, and therefore the inner product between any two of these joint states, $\braket{\psi_{i'j'}}{\psi_{ij}}_{ab}$, is known precisely. We can express an upper bound on \cref{eq:Y_ph_def_MDI} as the following optimization problem
		\begin{equation}
			\label{eq:sdp_mdi}
			\begin{aligned}
				&\max \textrm{ } Y_{Z\wedge\textrm{ph}}  \\
				&\textrm{s.t. } \ev{\hat M_{\gamma}^\dagger \hat M_{\gamma}} {\psi_{ij}}_{ab} =  Y_{i,j}^{\gamma} \quad \forall i,j, \gamma,\\
				&\sum_\gamma  \mel{\psi_{i'j'}}{\hat M_{\gamma}^\dagger \hat M_{\gamma}}{\psi_{ij}}_{ab} = \braket{\psi_{i'j'}}{\psi_{ij}}_{ab} \quad \forall i,i',j,j'.
			\end{aligned}
		\end{equation}
		
		Note that the first type of constraints in \cref{eq:sdp_mdi} hold due to \cref{eq:Y_ij_gamma}, and the second type of constraints hold due to \cref{eq:completeness_eq_mdi}. {As in the P\&M case, by defining $G$ as the Gram matrix of the vector set $\{\hat M_{\gamma} \ket{\psi_{ij}}_{ab}\}_{i,j,\gamma}$} {and substituting $Y_{Z\wedge\textrm{ph}}$ by its expression in \cref{eq:Y_ph_def_MDI}}, it is straightforward to see that \cref{eq:sdp_mdi} is an SDP.

		\subsection*{SDP with partial state characterization}
		
		In this case, the states $\{\ket{\psi_i}_a\}_i$ ($\{\ket{\psi_j}_b\}_j$) that Alice (Bob) emits are only partially characterized. This means that we only know a bound on their overlap with some other characterized states $\{\ket{\phi_i}_a\}_i$ ($\{\ket{\phi_j}_b\}_j$) (see  \cref{eq:assumption}). As in the P\&M case, without loss of generality, we can assume that these states have the form given by \cref{eq:psi_j_trick}. Taking the tensor product of \cref{eq:psi_j_trick} for both Alice's and Bob's states, we obtain the following expression for their joint states,
		\begin{equation}
			\label{eq:psi_ij_trick}
			\ket{\psi_{ij}}_{ab} \coloneqq \ket{\psi_{i}}_{a} \otimes \ket*{\psi_{j}}_{b} = \sqrt{1-\xi_{ij}} \ket*{\phi_{ij}}_{ab}+\sqrt{\xi_{ij}} \ket*{\phi_{ij}^\perp}_{ab}, 
		\end{equation}
		where we have defined $\xi_{ij} \coloneqq 1- (1-\epsilon_i)(1-\epsilon_j)$ and $\ket*{\phi_{ij}}_{ab} \coloneqq \ket*{\phi_{i}}_{a} \otimes \ket*{\phi_{j}}_{b}$. Also, $\ket*{\phi_{ij}^\perp}_{ab}$ is a normalized state that can be easily shown to satisfy $\braket*{\phi_{ij}^\perp}{\phi_{ij}}_{ab} = 0$. Then, we define our optimization problem as
		\begin{equation}
			\label{eq:sdp_mdi_2}
			\begin{aligned}
				&\max \textrm{ } Y_{Z\wedge\textrm{ph}}  \\
				&\textrm{s.t. } (1-\xi_{ij}) \mel*{\phi_{ij}}{\hat{M}_{\gamma}^\dagger \hat{M}_{\gamma}}{\phi_{ij}}_{ab} + \sqrt{\xi_{ij}(1-\xi_{ij})} \mel*{\phi_{ij}}{\hat{M}_{\gamma}^\dagger \hat{M}_{\gamma}}{\phi_{ij}^\perp}_{ab} \\
				&\quad \quad+ \sqrt{\xi_{ij}(1-\xi_{ij})} \mel*{\phi_{ij}^\perp}{\hat{M}_{\gamma}^\dagger \hat{M}_{\gamma}}{\phi_{ij}}_{ab} + \xi_{ij} \mel*{\phi_{ij}^\perp}{\hat{M}_{\gamma}^\dagger \hat{M}_{\gamma}}{\phi_{ij}^\perp}_{ab} =   Y_{i,j}^{\gamma} \quad \forall i,j, \gamma,\\
				&\sum_\gamma  \mel{\phi_{i'j'}}{\hat M_{\gamma}^\dagger \hat M_{\gamma}}{\phi_{ij}}_{ab} = \braket{\phi_{i'j'}}{\phi_{ij}}_{ab} \quad \forall i,i',j,j', \\
				&\sum_\gamma  \mel*{\phi_{ij}^\perp}{\hat M_{\gamma}^\dagger \hat M_{\gamma}}{\phi_{ij}}_{ab} = 0 \quad \forall j, \\
				&\sum_\gamma  \mel*{\phi_{ij}^\perp}{\hat M_{\gamma}^\dagger \hat M_{\gamma}}{\phi_{ij}^\perp}_{ab} = 1 \quad \forall j,
			\end{aligned}
		\end{equation}
		whose objective function is given by
		\begin{equation}
			\label{eq:Y_ph_def_MDI_expanded}
			\begin{aligned}
				Y_{Z \wedge \textrm{ph}} &= \frac{1}{8}\Big[
				(1-\xi_{00})\ev*{\hat M_{\tt pass}^\dagger \hat M_{\tt pass}}{\phi_{00}}_{ab} + \sqrt{\xi_{00}(1-\xi_{00})}\mel*{\phi_{00}}{\hat M_{\tt pass}^\dagger \hat M_{\tt pass}}{\phi_{00}^\perp}_{ab} \\
				&+ \sqrt{\xi_{00}(1-\xi_{00})}\mel*{\phi_{00}^\perp}{\hat M_{\tt pass}^\dagger \hat M_{\tt pass}}{\phi_{00}}_{ab} + \xi_{00}\ev*{\hat M_{\tt pass}^\dagger \hat M_{\tt pass}}{\phi_{00}^\perp}_{ab} \\
				&+ \sqrt{(1-\xi_{00})(1-\xi_{11})}\mel*{\phi_{00}}{\hat M_{\tt pass}^\dagger \hat M_{\tt pass}}{\phi_{11}}_{ab} + \sqrt{(1-\xi_{00})\xi_{11}}\mel*{\phi_{00}}{\hat M_{\tt pass}^\dagger \hat M_{\tt pass}}{\phi_{11}^\perp}_{ab} \\
				&+ \sqrt{\xi_{00}(1-\xi_{11})}\mel*{\phi_{00}^\perp}{\hat M_{\tt pass}^\dagger \hat M_{\tt pass}}{\phi_{11}}_{ab} + \sqrt{\xi_{00}\xi_{11}}\mel*{\phi_{00}^\perp}{\hat M_{\tt pass}^\dagger \hat M_{\tt pass}}{\phi_{11}^\perp}_{ab} \\
				&+ (1-\xi_{01})\ev*{\hat M_{\tt pass}^\dagger \hat M_{\tt pass}}{\phi_{01}}_{ab} + \sqrt{\xi_{01}(1-\xi_{01})}\mel*{\phi_{01}}{\hat M_{\tt pass}^\dagger \hat M_{\tt pass}}{\phi_{01}^\perp}_{ab} \\
				&+ \sqrt{\xi_{01}(1-\xi_{01})}\mel*{\phi_{01}^\perp}{\hat M_{\tt pass}^\dagger \hat M_{\tt pass}}{\phi_{01}}_{ab} + \xi_{01}\ev*{\hat M_{\tt pass}^\dagger \hat M_{\tt pass}}{\phi_{01}^\perp}_{ab} \\
				&+ \sqrt{(1-\xi_{01})(1-\xi_{10})}\mel*{\phi_{01}}{\hat M_{\tt pass}^\dagger \hat M_{\tt pass}}{\phi_{10}}_{ab} + \sqrt{(1-\xi_{01})\xi_{10}}\mel*{\phi_{01}}{\hat M_{\tt pass}^\dagger \hat M_{\tt pass}}{\phi_{10}^\perp}_{ab} \\
				&+ \sqrt{\xi_{01}(1-\xi_{10})}\mel*{\phi_{01}^\perp}{\hat M_{\tt pass}^\dagger \hat M_{\tt pass}}{\phi_{10}}_{ab} + \sqrt{\xi_{01}\xi_{10}}\mel*{\phi_{01}^\perp}{\hat M_{\tt pass}^\dagger \hat M_{\tt pass}}{\phi_{10}^\perp}_{ab} \\
				&+ \sqrt{(1-\xi_{10})(1-\xi_{01})}\mel*{\phi_{10}}{\hat M_{\tt pass}^\dagger \hat M_{\tt pass}}{\phi_{01}}_{ab} + \sqrt{(1-\xi_{10})\xi_{01}}\mel*{\phi_{10}}{\hat M_{\tt pass}^\dagger \hat M_{\tt pass}}{\phi_{01}^\perp}_{ab} \\
				&+ \sqrt{\xi_{10}(1-\xi_{01})}\mel*{\phi_{10}^\perp}{\hat M_{\tt pass}^\dagger \hat M_{\tt pass}}{\phi_{01}}_{ab} + \sqrt{\xi_{10}\xi_{01}}\mel*{\phi_{10}^\perp}{\hat M_{\tt pass}^\dagger \hat M_{\tt pass}}{\phi_{01}^\perp}_{ab} \\
				&+ (1-\xi_{10})\ev*{\hat M_{\tt pass}^\dagger \hat M_{\tt pass}}{\phi_{10}}_{ab} + \sqrt{\xi_{10}(1-\xi_{10})}\mel*{\phi_{10}}{\hat M_{\tt pass}^\dagger \hat M_{\tt pass}}{\phi_{10}^\perp}_{ab} \\
				&+ \sqrt{\xi_{10}(1-\xi_{10})}\mel*{\phi_{10}^\perp}{\hat M_{\tt pass}^\dagger \hat M_{\tt pass}}{\phi_{10}}_{ab} + \xi_{10}\ev*{\hat M_{\tt pass}^\dagger \hat M_{\tt pass}}{\phi_{10}^\perp}_{ab} \\
				&+ \sqrt{(1-\xi_{11})(1-\xi_{00})}\mel*{\phi_{11}}{\hat M_{\tt pass}^\dagger \hat M_{\tt pass}}{\phi_{00}}_{ab} + \sqrt{(1-\xi_{11})\xi_{00}}\mel*{\phi_{11}}{\hat M_{\tt pass}^\dagger \hat M_{\tt pass}}{\phi_{00}^\perp}_{ab} \\
				&+ \sqrt{\xi_{11}(1-\xi_{00})}\mel*{\phi_{11}^\perp}{\hat M_{\tt pass}^\dagger \hat M_{\tt pass}}{\phi_{00}}_{ab} + \sqrt{\xi_{11}\xi_{00}}\mel*{\phi_{11}^\perp}{\hat M_{\tt pass}^\dagger \hat M_{\tt pass}}{\phi_{00}^\perp}_{ab} \\
				&+ (1-\xi_{11})\ev*{\hat M_{\tt pass}^\dagger \hat M_{\tt pass}}{\phi_{11}}_{ab} + \sqrt{\xi_{11}(1-\xi_{11})}\mel*{\phi_{11}}{\hat M_{\tt pass}^\dagger \hat M_{\tt pass}}{\phi_{11}^\perp}_{ab} \\
				&+ \sqrt{\xi_{11}(1-\xi_{11})}\mel*{\phi_{11}^\perp}{\hat M_{\tt pass}^\dagger \hat M_{\tt pass}}{\phi_{11}}_{ab} + \xi_{11}\ev*{\hat M_{\tt pass}^\dagger \hat M_{\tt pass}}{\phi_{11}^\perp}_{ab}
				\Big].
			\end{aligned}
		\end{equation}
		To obtain the first type of constraints in \cref{eq:sdp_mdi_2}, we have substituted \cref{eq:psi_ij_trick} into \cref{eq:Y_ij_gamma}, and the second, third and fourth type of constraints hold due to \cref{eq:completeness_eq_mdi}.  Also, to obtain the objective function in \cref{eq:Y_ph_def_MDI_expanded}, we have substituted \cref{eq:psi_ij_trick} into \cref{eq:Y_ph_def_MDI}. As in the P\&M case, by defining $G$ as the Gram matrix of the union of the vector  sets $\{\hat M_\gamma \ket{\phi_{ij}}_{ab}\}_{i,j,\gamma}$ and $\{\hat M_\gamma \ket*{\phi_{ij}^\perp}_{ab}\}_{i,j,\gamma}$, it is straightforward to see that \cref{eq:sdp_mdi_2} is an SDP.

		\section{Application to decoy-state protocols}
		\label{app:decoy_state}

		In this Appendix, we extend our analysis to scenarios in which Alice's source emits mixed states with partially characterized eigenvectors, which enables its application to decoy-state protocols. For such protocols, we can incorporate not only imperfections and side channels in the encoding of bit-and-basis information, but also vulnerabilities in the implementation of the decoy-state method itself, such as information leakage of Alice's intensity setting choices. 
		
		{For simplicity, in this Appendix, we focus only on decoy-state P$\&$M protocols. However, our framework can also be applied to decoy-state MDI-type schemes by combining the ideas presented here with the analysis in \cref{app:MDI}. This enables the application of our techniques to the standard decoy-state MDI-QKD protocol \cite{loMeasurementDeviceIndependentQuantum2012}, to twin-field QKD protocols \cite{wang2019twinfield,curtySimpleSecurity2019} and to mode-pairing QKD \cite{zengModepairingQuantum2022,xieBreakingRateLoss2022}.}      
		
		First, in \cref{appsec:decoy_state_general_analysis},
		we present the general version of our approach; then, in \cref{appsec:THA_analysis},
		we particularize it to derive secret-key rates for a decoy-state {BB84} protocol in the presence of a THA against both the intensity modulator and the bit-and-basis encoder. The simulation results for the latter scenario are shown in \cref{subsec:decoy_THA_results}.
		
		\subsection{General analysis}
		\label{appsec:decoy_state_general_analysis}

		Consider a general decoy-state-type scenario in which, when Alice chooses the bit-and-basis setting $j$ and the intensity setting $\mu$, she emits some mixed state $\rho_{j,\mu}$. This {state can be understood as the average over some parameter that is (at least partially) outside Eve's control and knowledge, such as the phase of a coherent state that has undergone a phase-randomization process. A mixed state} always admits an eigendecomposition of the form
		\begin{equation}
			\label{eq:eigendecomp_general_decoy}
			\rho_{j,\mu} = \sum_n p_{n\vert j,\mu} \ketbra*{\psi_{j,\mu}^{(n)}},
		\end{equation}
		where $p_{n \vert j,\mu}$ can be interpreted as the probability to prepare the state $\ket*{\psi_{j,\mu}^{(n)}}$ and $\braket*{\psi_{j,\mu}^{(n')}}{\psi_{j,\mu}^{(n)}} = \delta_{nn'}$.
		Here, we show how to use our analysis to prove security for this type of scenario when the eigenvectors $\ket*{\psi_{j,\mu}^{(n)}}$ are only partially characterized, i.e., they are known to satisfy
		\begin{equation}
			\label{eqn:assumption_general_decoy}
			\abs*{\braket*{\phi_{j,\mu}^{(n)}}{\psi_{j,\mu}^{(n)}}}^2 \geq 1 - \epsilon_{j,\mu}^{(n)},
		\end{equation}
		for known $\ket*{\phi_{j,\mu}^{(n)}}$ and known $\epsilon_{j,\mu}^{(n)}$. For simplicity, in our derivations below, we consider the restricted scenario in which the eigenvalues in \cref{eq:eigendecomp_general_decoy} only depend on $\mu$ but not on $j$, i.e.,
		\begin{equation}
			\label{eq:eigendecomp_general_decoy_2}
			\rho_{j,\mu} = \sum_n p_{n\vert \mu} \ketbra*{\psi_{j,\mu}^{(n)}},
		\end{equation}
		and that these eigenvalues are known. This is the case, for instance, when Alice emits phase-randomized coherent states with known phase distributions and known intensities. However, we remark that our analysis can be applied even if the eigenvalues depend on $j$ and are not known, as long as one has a lower bound on them\footnote{Our analysis can be applied as long as one has a lower bound $p_{n\vert \mu}^{\rm L}$ on $p_{n\vert j,\mu}$, by trivially substituting $p_{n\vert \mu} \to p_{n\vert \mu}^{\rm L}$ in \cref{eq:asymp_SKR_decoy,eq:asymp_SKR_decoy_n1,eq:sdp_yield}.}.

		As shown in \cref{app:proof_emitted_states}, \cref{eqn:assumption_general_decoy} implies that the states $\ket{\psi_{j,\mu}}$ can be assumed to have the form
		\begin{equation}
			\label{eq:psi_j_n_trick}
			\ket*{\psi_{j,\mu}^{(n)}} = \sqrt{1-\epsilon_{j,\mu}^{(n)}} \ket*{\phi_{j,\mu}^{(n)}} + \sqrt{\epsilon_{j,\mu}^{(n)}} \ket*{\phi_{j,\mu}^{(n),\perp}},
		\end{equation}
		where $\ket*{\phi_{j,\mu}^{(n),\perp}}$ is a state orthogonal to $\ket*{\phi_{j,\mu}^{(n)}}$, i.e., it satisfies $\braket*{\phi_{j,\mu}^{(n),\perp}}{\phi_{j,\mu}^{(n)}} = 0$.
		
		For concreteness, we will assume a decoy-state protocol in which Alice and Bob extract their sifted key from the detected events in which both users chose the $Z$ basis and Alice chose a particular intensity setting, say $\mu_0$. A lower bound on the asymptotic secret-key rate $R$ can then be expressed as
		\begin{equation}
			\label{eq:asymp_SKR_decoy}
			R \geq \sum_n p_{n \vert \mu_0} Y_{Z,\mu_0}^{(n)} \big[1-h(e_{\textrm{ph},\mu_0}^{(n)})-fh(e_{Z,\mu_0})\big].
		\end{equation}
		Here, $e_{Z,\mu_0}$ refers to the observed error rate of the sifted key, $Y_{Z,\mu_0}^{(n)}$ denotes the probability that Bob obtains a conclusive bit outcome when Alice emits $\ket*{\psi_{0,\mu_0}^{(n)}}$ or $\ket*{\psi_{1,\mu_0}^{(n)}}$, and $e_{\textrm{ph},\mu_0}^{(n)}$ represents the phase-error rate associated to these latter events. 
		
		To obtain a lower bound on the asymptotic secret-key rate in \cref{eq:asymp_SKR_decoy}, one needs to find a lower bound on $Y_{Z,\mu_0}^{(n)}$ and an upper bound on $e_{\textrm{ph},\mu_0}^{(n)}$ for some values of $n$. Here, for simplicity, we focus on the typical approach in which one extracts secret key only from the states with $n=1$, in which case the lower bound on the secret-key rate becomes
		\begin{equation}
			\label{eq:asymp_SKR_decoy_n1}
			R \geq p_{1 \vert \mu_0} Y_{Z,\mu_0}^{(1)} \big[1-h(e_{\textrm{ph},\mu_0}^{(1)})-fh(e_{Z,\mu_0})\big].
		\end{equation}
		However, we remark that our analysis below can be used to find bounds for any $n$ {by simply changing the superscript $(1)$ to $(n)$ in the objective functions of \cref{eq:sdp_yield,eq:Y_Zmu0ph}}.
		
		A lower bound  $Y_{Z,\mu_0}^{(1),\textrm{L}}$ on $Y_{Z,\mu_0}^{(1)}$ can be obtained by solving the following optimization problem
		\begin{equation}
			\label{eq:sdp_yield}
			\begin{aligned}
				\min \textrm{ }& Y_{Z,\mu_0}^{(1)} := 1-\frac{1}{2}\sum_{j\in\{0,1\}}\ev*{\hat M_{f}^\dagger \hat M_{f}}{\psi_{j,\mu_0}^{(1)}}  \\
				&\textrm{s.t. } \sum_{n=0}^{n_{\rm cut}} p_{n\vert\mu}\ev*{\hat M_{\gamma}^\dagger \hat M_{\gamma}} {\psi_{j,\mu}^{(n)}} \leq  Q_{j,\mu}^{\gamma} \quad \forall j,\mu, \gamma,\\
				&\sum_{n=0}^{n_{\rm cut}} p_{n\vert\mu}\big(1-\ev*{\hat M_{\gamma}^\dagger \hat M_{\gamma}} {\psi_{j,\mu}^{(n)}}\big) \leq  1-Q_{j,\mu}^{\gamma} \quad \forall j,\mu, \gamma,\\
				&\sum_\gamma  \mel*{\phi_{j',\mu'}^{(n')}}{\hat M_{\gamma}^\dagger \hat M_{\gamma}}{\phi_{j,\mu}^{(n)}} = \braket*{\phi_{j',\mu'}^{(n')}}{\phi_{j,\mu}^{(n)}} \quad \forall j,j',\mu,\mu', \forall n,n' \in \{0,1,...,n_\mathrm{cut}\}, \\
				&\sum_\gamma  \mel*{\phi_{j,\mu}^{(n),\perp}}{\hat M_{\gamma}^\dagger \hat M_{\gamma}}{\phi_{j,\mu}^{(n)}} = 0 \quad \forall j,\mu, \forall n \in \{0,1,...,n_\mathrm{cut}\}, \\
				&\sum_\gamma  \ev*{\hat M_{\gamma}^\dagger \hat M_{\gamma}}{\phi_{j,\mu}^{(n),\perp}} = 1 \quad \forall j,\mu, \forall n \in \{0,1,...,n_\mathrm{cut}\}. 
			\end{aligned}
		\end{equation}
		Here, $Q_{j,\mu_0}^{\gamma}$ is the observed rate at which Bob obtains the outcome $\gamma \in \{0_X,1_X,f\}$ when Bob chooses the $X$ basis\footnote{Note that, due to the basis-independent detection efficiency assumption, one can indistinctly use either Bob's $Z$-basis measurement results or Bob's $X$-basis measurement results when defining the SDP to estimate $Y_{Z,\mu_0}^{(1)}$. When estimating $Y_{Z,\mu_0\wedge\mathrm{ph}}^{(1),\textrm{U}}$, one needs to define the SDP using Bob's $X$-basis measurement results. Here, for simplicity, we use Bob's $X$-basis measurement results in both SDPs.} and Alice chooses setting $j$ and intensity $\mu$.  Similar to the case of discrete-modulated protocols in \cref{sec:main_result}, by substituting the expression of $\ket*{\psi_{j,\mu}^{(n)}}$ in \cref{eq:psi_j_n_trick} into \cref{eq:sdp_yield}, and defining $G$ as the Gram matrix of the union of the vector sets $\{\hat M_\gamma \ket*{\phi_{j,\mu}^{(n)}}\}_{j,\mu,\gamma,n}$ and $\{\hat M_\gamma \ket*{\phi_{j,\mu}^{(n),\perp}}\}_{j,\mu,\gamma,n}$, it is easy to see that \cref{eq:sdp_yield} is an SDP.
		
		The phase-error rate is defined as $e_{\rm{ph},\mu_0}^{(1)}=Y_{Z,\mu_0\wedge\mathrm{ph}}^{(1)}/Y_{Z,\mu_0}^{(1)}$, where  $Y_{Z,\mu_0\wedge\mathrm{ph}}^{(1)}$ is defined by substituting $\ket{\psi_j} \to \ket*{\psi_{j,\mu_0}^{(1)}}$ in \cref{eq:Y_ph_def_PM}, i.e.,
		\begin{equation}
			\label{eq:Y_Zmu0ph}
			\begin{split}
				Y_{Z,\mu_0\wedge\mathrm{ph}}^{(1)}
				=& \frac{1}{4} \bigg[\ev*{\hat{M}_{1_X}^\dagger\hat{M}_{1_X}}{\psi_{0,\mu_0}^{(1)}}
				+ \mel*{\psi_{0,\mu_0}^{(1)}}{\hat{M}_{1_X}^\dagger\hat{M}_{1_X}}{\psi_{1,\mu_0}^{(1)}}
				+ \mel*{\psi_{1,\mu_0}^{(1)}}{\hat{M}_{1_X}^\dagger\hat{M}_{1_X}}{\psi_{0,\mu_0}^{(1)}}
				+ \ev*{\hat{M}_{1_X}^\dagger \hat{M}_{1_X}}{\psi_{1,\mu_0}^{(1)}}  \\ 
				&+\ev*{\hat{M}_{0_X}^\dagger\hat{M}_{0_X}}{\psi_{0,\mu_0}^{(1)}}
				- \mel*{\psi_{0,\mu_0}^{(1)}}{\hat{M}_{0_X}^\dagger\hat{M}_{0_X}}{\psi_{1,\mu_0}^{(1)}}
				- \mel*{\psi_{1,\mu_0}^{(1)}}{\hat{M}_{0_X}^\dagger\hat{M}_{0_X}}{\psi_{0,\mu_0}^{(1)}}
				+ \ev*{\hat{M}_{0_X}^\dagger \hat{M}_{0_X}}{\psi_{1,\mu_0}^{(1)}}  \bigg].
			\end{split}
		\end{equation}
		
		An upper bound $Y_{Z,\mu_0\wedge\mathrm{ph}}^{(1),\textrm{U}}$ on $Y_{Z,\mu_0\wedge\mathrm{ph}}^{(1)}$ can be obtained simply by replacing the objective function of the SDP in \cref{eq:sdp_yield} by $\max \textrm{ } Y_{Z,\mu_0\wedge\mathrm{ph}}^{(1)}$, and substituting the expression of $\ket*{\psi_{j,\mu}^{(n)}}$ in \cref{eq:psi_j_n_trick} into the definition of $Y_{Z,\mu_0\wedge\mathrm{ph}}^{(1)}$ in \cref{eq:Y_Zmu0ph}. Then, an upper bound on the phase-error rate is given by
		\begin{equation}
			e_{\rm{ph},\mu_0}^{(1)} \leq \frac{Y_{Z,\mu_0\wedge\mathrm{ph}}^{(1),\textrm{U}}}{Y_{Z,\mu_0}^{(1),\textrm{L}}}.
		\end{equation}
		Thus, a lower bound on the secret-key rate in \cref{eq:asymp_SKR_decoy_n1} can be obtained by simply solving two SDPs.
		
		\subsection{Analysis of a THA against the intensity modulator and the bit-and-basis encoder}
		\label{appsec:THA_analysis}
		
		As an example to illustrate the applicability of our analysis to realistic scenarios, here, we consider a practical decoy-state case in which Eve performs a THA against the intensity modulator and the bit-and-basis encoder, which allows her to learn information about Alice's setting choices $j$ and $\mu$. For concreteness, we consider a standard decoy-state BB84-type protocol in which Eve's THA is the only imperfection present, although, as already discussed, our analysis could incorporate other types of flaws too, such as SPFs.
		Remarkably, the only information needed to apply our analysis to this scenario is an upper bound $I_{\rm max}$ on the intensity of the back-reflected light induced by Eve's THA. That is, no information about the state of the injected light or the back-reflected light is required.
		
		In this situation, the full emitted states when Alice chooses $j$ and $\mu$ can be expressed as
		\begin{equation}
			\label{eq:rho_j_mu_THA}
			\rho_{j,\mu} = \rho_{j,\mu}^{\rm ideal} \otimes \ketbra{\lambda_{j,\mu}}_E,
		\end{equation}
		where%
		\begin{equation}
			\label{eq:ideal_PRWCP}
			\rho_{j,\mu}^{\rm ideal} = \sum_{n=0}^{\infty}  p_{n \vert \mu} \ketbra{n_j} 
		\end{equation}
		is an ideal phase-randomized weak coherent pulse, with $\ket{n_j}$ being an ideal $n$-photon state with bit-and-basis encoding $j$ and 
		\begin{equation}
			p_{n|\mu}=e^{-\mu}\frac{\mu^n}{n!}
		\end{equation}
		follows a Poisson distribution. Also, in \cref{eq:rho_j_mu_THA}, $\ket{\lambda_{j,\mu}}_E$ is the state of the back-reflected light available to Eve, which contains information about $j$ and $\mu$, and which can be assumed to be pure without loss of generality. The specific quantum state of $\ket{\lambda_{j,\mu}}_E$ is uncharacterized but, by assumption, we have an upper bound $I_{\rm max}$ on its mean photon number. Using this fact, one can show that (see Ref.~\cite[Appendix C]{pereiraModifiedBB842023})
		\begin{equation}
			\label{eq:vacuum_bound_imax}
			\abs*{\braket{\mathrm{vac}}{\lambda_{j,\mu}}_E}^2 \geq 1-I_{\rm max},
		\end{equation}
		where {$\ket{\mathrm{vac}}_E$} is the vacuum state of the back-reflected system $E$, which cannot contain any information about Alice's setting choices. As shown in \cref{app:proof_emitted_states}, \cref{eq:vacuum_bound_imax} implies that, without loss of generality, the state of the back-reflected light can be assumed to have the following form
		\begin{equation}
			\label{eq:lambda_assumed_form}
			\ket{\lambda_{j,\mu}}_E = \sqrt{1-I_{\rm max}} \ket{\mathrm{vac}}_E + \sqrt{I_{\rm max}} \ket{\Omega_{j,\mu}}_E,
		\end{equation}
		where $\ket{\Omega_{j,\mu}}_E$ is a state such that $\braket{\mathrm{vac}}{\Omega_{j,\mu}}_E = 0$. Importantly, this state contains photons and therefore may leak information about Alice's setting choices $j$ and $\mu$. 
		
		Using \cref{eq:ideal_PRWCP,eq:lambda_assumed_form}, we can rewrite the state $\rho_{j,\mu}$ in \cref{eq:rho_j_mu_THA} as given in \cref{eq:eigendecomp_general_decoy_2}. For this, we define
		\begin{equation}
			\ket*{\psi_{j,\mu}^{(n)}} = \sqrt{1-I_{\rm max}} \ket*{\phi_{j}^{(n)}} + \sqrt{I_{\rm max}} \ket*{\phi_{j,\mu}^{(n),\perp}}
		\end{equation}
		with
		\begin{equation}
			\begin{aligned}
				\ket*{\phi_{j}^{(n)}} &= \ket{n_j} \otimes \ket{\mathrm{vac}}_E, \\
				\ket*{\phi_{j,\mu}^{(n),\perp}} &= \ket{n_j} \otimes \ket{\Omega_{j,\mu}}_E,
			\end{aligned}
		\end{equation}
		and $\braket*{\phi_{j,\mu}^{(n),\perp}}{\phi_{j}^{(n)}} = 0$. This implies that one can apply the analysis in \cref{appsec:decoy_state_general_analysis} to evaluate the secret-key rate obtainable in this scenario. In \cref{subsec:decoy_THA_results}, we do so for different values of $I_{\rm max}$.

		\section{Justification of \cref{eq:psi_j_trick}}
		\label{app:proof_emitted_states}
		
		{Here, we provide a rigorous justification for our assumption that, in the case of partial state characterization, Alice's states have the form given by \cref{eq:psi_j_trick}. While this proof has been originally introduced in Ref.~\cite[Supplementary Methods 1]{curras-lorenzoSecurityFramework2023}, we reproduce it here for completeness.
			
			\begin{lemma}[\cite{curras-lorenzoSecurityFramework2023}]
				Consider a QKD protocol in which Alice prepares some states $\{\ket{\psi_j}_a\}_j$, with $j$ representing her setting choice. Assume that these states satisfy
				\begin{equation}
					\label{eqapp:assumption}
					\abs{\braket{\phi_j}{\psi_{j}}_a}^{{2}} \geq 1 - \epsilon_j,  
				\end{equation}
				where $\{\ket{\phi_j}_a\}_j$ is another set of states. Then, when proving security, Alice's emitted states can be assumed to have the form
				\begin{equation}
					\label{eqapp:objective}
					\ket{\psi_j}_a = \sqrt{1-\epsilon_j} \ket{\phi_j}_a + \sqrt{\epsilon_j} \ket*{\phi_j^\perp}_a,
				\end{equation}
				where $\ket*{\phi_j^\perp}_a$ is a state such that $\braket*{\phi_j^\perp}{\phi_j}_a = 0$, since this never underestimates the information available to Eve.
			\end{lemma}
			\begin{proof}
				Without loss of generality, \cref{eqapp:assumption} directly implies that the emitted states can be expressed as
				\begin{equation}
					\ket{\psi_j}_a = e^{i \varphi_j}(\sqrt{1-\epsilon'_j} \ket{\phi_j}_a + \sqrt{\epsilon'_j} \ket*{\phi_j^\perp}_a){,}
					\label{eq:psi_phase}
				\end{equation}
				%
				where $0\leq \epsilon'_j \leq \epsilon_j$, $\varphi_j \in [0,2\pi)$ and $\ket*{\phi_j^\perp}_a$ is a state such that $\braket*{\phi_j^\perp}{\phi_j}_a = 0$. However, since the global phases $\varphi_j$ have no physical meaning, we can consider that $\varphi_j = 0$ and therefore
				\begin{equation}
					\label{eqapp:emitted_states_assumed_form}
					{\ket{\psi_j}_a =\sqrt{1-\epsilon'_j} \ket{\phi_j}_a + \sqrt{\epsilon'_j} \ket*{\phi_j^\perp}_a.}
				\end{equation}
				{Next, we show that one can simply consider that $\epsilon'_j = \epsilon_j$, since this represents an advantageous scenario for Eve.} To see this, consider a fictitious scenario in which,  instead of {emitting} the states $\{\ket{\psi_j}_{a}\}_j$ in \cref{eqapp:emitted_states_assumed_form}, Alice emits the states $\{\ket*{\psi_j}_{aF}\}_{j}$ defined as
				\begin{equation}
					\label{eqapp:psi_prime}
					\begin{gathered}
						\ket*{\psi_j}_{aF} =	\ket{\psi_j}_a \otimes \Bigg[\sqrt{\frac{1-\epsilon_j}{1-\epsilon'_j}} \ket{0}_F +  {\sqrt{1-\frac{1-\epsilon_j}{1-\epsilon'_j}}} \ket{1}_F \Bigg],\\
					\end{gathered}
				\end{equation}
				with $\{\ket{0}_F,\ket{1}_F\}$ forming an {orthonormal} basis for the fictitious system $F$. This fictitious scenario cannot be disadvantageous for Eve, as she could simply ignore system $F$ and execute the same attack as she would have executed for the actual protocol. Importantly, we can rewrite \cref{eqapp:psi_prime} as
				\begin{align}
					\ket*{\psi_j}_{aF} = \sqrt{1-\epsilon_j} \ket{\phi_j}_{aF} + \sqrt{\epsilon_j} \ket*{\phi_j^\perp}_{aF},
				\end{align}
				where we have defined the normalized states $\ket{\phi_j}_{aF} \coloneqq \ket{\phi_j}_{a} \ket{0}_F$ and 
				\begin{equation}
					\ket*{\phi_j^\perp}_{aF} = \frac{\sqrt{\epsilon_j - \epsilon'_j}}{\sqrt{\epsilon_j}} \ket{\phi_j}_a\ket{1}_F 
					+ \frac{\sqrt{\epsilon'_j}}{\sqrt{\epsilon_j}}  \ket*{\phi_j^\perp}_a \otimes \left[ \sqrt{\frac{1-\epsilon_j}{1-\epsilon'_j}} \ket{0}_F 
					+ \sqrt{\frac{\epsilon_j - \epsilon'_j}{1-\epsilon'_j}} \ket{1}_F \right].
				\end{equation}
				
				It is easy to see that this construction mantains the required orthogonality condition $\braket*{\phi_j^\perp}{\phi_j}_{aF} = 0$. Therefore, by renaming systems $aF$ as simply $a$, we obtain \cref{eqapp:objective}, {as desired.} {The advantage of \cref{eqapp:objective} with respect to \cref{eqapp:emitted_states_assumed_form} is that, in the former,} the coefficient of the state $\ket{\phi_j}_{a}$ is known exactly, which is needed to run the SDP in \cref{eq:sdp_2}.
		\end{proof}}

	\end{document}